\documentclass[fontsize=11pt, headings=normal, headinclude=true, paper=a4, pagesize, cleardoublepage=current, BCOR=10mm, fleqn, numbers=noenddot]{scrartcl}
\pdfoutput=1
\usepackage{lmodern}

\setkomafont{sectioning}{\normalcolor\bfseries}
\recalctypearea 
\setcapindent{0pt}

  \usepackage[english]{babel}  
  \usepackage[numbers,square]{natbib}
  \usepackage{mathrsfs}
  \usepackage{amsmath}
  \usepackage{amssymb}
  \usepackage{amsthm}
  \usepackage{enumitem}
  \usepackage[utf8]{inputenc}    
  \usepackage[T1]{fontenc}         
  \usepackage{graphicx}            
  \usepackage{color}		
  \usepackage[english=british]{csquotes} 
  
  \usepackage{tikz}

\usetikzlibrary{matrix}

\newtheorem{thm}{Theorem}[section]
\newtheorem{prop}[thm]{Proposition}
\newtheorem{lem}[thm]{Lemma}
\newtheorem{cor}[thm]{Corollary}

\newtheorem*{thm*}{Theorem}

\theoremstyle{definition}
\newtheorem{definition}[thm]{Definition}
\newtheorem{cond}[thm]{Condition}

\newtheorem{rem}[thm]{Remark}

\renewcommand{\phi}{\varphi}
\newcommand{\eps}{\varepsilon}

\newcommand{\ud}{\mathrm{d}}
\newcommand{\uod}{{\mathrm{od}}}
\newcommand{\ue}{\mathrm{e}}
\newcommand{\ui}{\mathrm{i}}
\newcommand{\R}{\mathbb{R}}

\newcommand{\N}{\mathbb{N}}
\newcommand{\Fou}{\mathcal{F}}
\newcommand{\Lz}{L^2}

\newcommand{\norm}[1]{\ensuremath{\left\lVert #1 \right\rVert}}
\newcommand{\abs}[1]{\ensuremath{\left\lvert #1 \right\rvert}}

\newcommand{\hilb}{\mathscr{H}}
\newcommand{\uppar}[1]{\ensuremath{^{(#1)}}}

\title{On Nelson-type Hamiltonians and abstract boundary conditions}

\author{Jonas Lampart 
\thanks{CNRS \& ICB (UMR 6303), Université de Bourgogne Franche-Comté, 9 Av. A. Savary, 21078 Dijon Cedex, France.
\texttt{jonas.lampart@u-bourgogne.fr}}
, Julian Schmidt 
\thanks{Fachbereich Mathematik, Eberhard Karls Universität Tübingen, Auf der Morgenstelle 10, 72076 Tübingen, Germany. \texttt{juls@maphy.uni-tuebingen.de}}
}

\begin{document}

\maketitle  

\begin{abstract}
We construct Hamiltonians for systems of nonrelativistic particles linearly coupled to massive scalar bosons using abstract boundary conditions. The construction yields an explicit characterisation of the domain of self-adjointness in terms of boundary conditions that relate sectors with different numbers of bosons. We treat both models in which the Hamiltonian may be defined as a form perturbation of the free operator, such as Fröhlich's polaron, and renormalisable models, such as the massive Nelson model. 
\end{abstract}

\section{Introduction}

We consider a system of nonrelativistic particles interacting with massive scalar bosons. For a linear coupling, the interaction between one particle and the bosons is (formally) given by $a(v(x-y))+a^*(v(x-y))$, where $a, a^*$ are the bosonic annihilation and creation operators, $v$ is the form factor of the interaction and $x$ denotes the position of the particle, $y$ that of a boson. Figuratively speaking, the particles act as sources that create and annihilate bosons with wavefunction $v$ centred at their position $x$.
We will discuss a class of ultraviolet-divergent models for which $v(y)$ is a singular function (or a distribution). In most examples $v(y)$ is singular at $y=0$ but regular and decaying as $|y|\to \infty$. For example, for the Fröhlich polaron $v(y)\sim |y|^{-2}$, and in the Nelson model $v(y)\sim |y|^{-5/2}$ (both in three space-dimensions). The Hamiltonians for these models can be constructed using quadratic forms (for the Fröhlich model) or by a renormalisation procedure (for the Nelson model). However, these methods do not give detailed and explicit information on the domain of the operator (e.g. concerning regularity) or the action of the operator thereon.  
We will discuss a new method of construction that explicitly describes the domain in terms of abstract boundary conditions relating sectors with different numbers of bosons. 
More precisely, the elements of the domain will, for any given number $n\geq 1$ of bosons, be singular functions with singularities determined by the function with $n-1$ bosons.
If the only singularity of $v$ is at $y=0$, these singularities are located on the planes in configuration space where the positions of (at least) a source and a boson coincide.
The relation between the form of this singularity and the function with fewer bosons can be viewed as an inhomogeneous  generalised  boundary condition on the set of these planes.

Boundary conditions of this type were proposed as an approach to ultraviolet divergences by Teufel and Tumulka~\cite{TeTu15, TeTu16}. They were called interior-boundary conditions, as they concern points in the interior of the configuration space of the two species of particles. 
Similar boundary conditions had previously been investigated by Thomas~\cite{thomas1984} in a specific model where the total number of particles is at most three.
The emphasis of these works is on point interactions, where $v$ is the $\delta$-distribution and it is particularly natural to consider boundary conditions.
A rigorous analysis of a model for nonrelativistic bosons, with $v=\delta$ and sources that are fixed at points in $\R^3$, was subsequently performed by Teufel, Tumulka, and the authors~\cite{IBCpaper}.  This extended a result of Yafaev~\cite{yafaev1992}, allowing only for the creation of a single particle.
The one-dimensional variant of this model was studied by Keppeler and Sieber~\cite{KeSi16}.

In the present article, we will explain how such an approach can be applied to models for nonrelativistic particles interacting with bosons, where the \enquote{sources} are themselves dynamical objects. We also demonstrate that the method is sufficiently flexible to accommodate various interactions $v$ and dispersion relations of the bosons, such as the relativistic dispersion of the Nelson model. Our class of models also contains a dynamical version of the model with nonrelativistic bosons and $v=\delta$ of~\cite{IBCpaper} in two (instead of three) space-dimensions.  Our method could also be applied to models that involve creation and annihilation of fermions, but we will restrict ourselves to bosons in this article.
We obtain an explicit characterisation of the Hamiltonian and its domain of self-adjointness, which seems to be new for all of the cases under consideration.
We also hope that this explicit characterisation will facilitate further research on the properties of these models, such as their energy-momentum spectrum and dynamics, which is an active area of investigation (see e.g.~\cite{AmFa14, BlTh17, GHL14, Miy18, MaMo17} for some recent results, and references therein).

\subsection{Nonrelativistic particles interacting with scalar bosons}

Let us now introduce some notation and discuss in more detail the models we will consider as well as our main results. We  consider a fixed but arbitrary number $M$ of nonrelativistic particles in $d\leq 3$ dimensions interacting with a variable number of scalar bosons. 
We do not impose any particular symmetry under permutations on the first type of particles.
The Hilbert space on which we describe our system is given by
\begin{equation*}
 \hilb:=L^2(\R^{dM})\otimes \Gamma(L^2(\R^d))
 =\bigoplus_{n = 0}^\infty \Lz(\R^{dM}) \otimes \Lz_{\mathrm{sym}}(\R^{dn}) = \bigoplus_{n \in \N} \hilb \uppar n 
 \,,
\end{equation*}
where $\Gamma(L^2(\R^d))$ is the bosonic Fock space over $L^2(\R^d)$ and $\hilb\uppar{n}$ the sector of $\hilb$ with $n$ bosons.
In the position representation, we will denote the positions of the first type of particles by $x_1,\dots, x_M$ and refer to these as the $x$-particles from now on. We will denote the positions of the bosons by $y_1, \dots$ and refer to them as the $y$-particles.
In appropriate units, the formal expression for the linearly coupled Hamiltonian of this system reads
\begin{equation}\label{eq:H formal}
 - \sum_{j=1}^M \Delta_{x_j} + \ud \Gamma(\omega(-\ui\nabla_y)) + g\sum_{j=1}^M \left( a^*(v(x_j-y))+a(v(x_j-y)) \right)\,,
\end{equation}
where $\omega:\R^d\to \R_+$ is the dispersion relation of the bosons, $v\in \mathscr{S}'(\R^d)$ is the interaction, and $g\in\R$ is the coupling constant.
When $v\in L^2$ and $\omega(k)\geq e_0>0$, then, by the Kato-Rellich theorem, this defines a self-adjoint operator on the domain
\begin{equation*}
D(L)=\{\psi\in\hilb : L\psi \in \hilb\} 
\end{equation*}
of the free operator (understood in the sense of tempered distributions)
\begin{equation}\label{eq:Ldef}
 L:=- \sum_{j=1}^M \Delta_{x_j} + \ud \Gamma(\omega(-\ui\nabla_y))\,.
\end{equation}
Note that $D(L)$ is contained in the domain of the boson-number operator $N=\ud \Gamma(1)$ if $\omega(k)\geq e_0>0$.

Our class of models concerns cases where the operator in Equation~\eqref{eq:H formal} above is not immediately well defined because $v\notin L^2(\R^d)$. We will only consider cases with an ultraviolet problem but no infrared problem, that is $\omega(k)\geq e_0>0$ and $\hat v \in L^2_\mathrm{loc}$.
The problem in this case is that the creation operator $a^*(v(x-y))$ is not a densely defined operator on $\hilb$, so the expression~\eqref{eq:H formal} cannot be interpreted as a sum of unbounded operators on any dense domain.
The annihilation operator $\sum_{j=1}^M a(v(x_j-y))$ is less problematic, as it is always densely defined, and under our assumptions it is defined on $D(L)$ (cf.~Corollary~\ref{lem:gonfockspace} and the following remark). Depending on $v$ and $\omega$, this problem may be solvable by one of two well-known methods.
\begin{enumerate}[label=(\arabic*)]
 \item If $\int \frac{ |\hat v(k)|^2}{k^2+\omega(k)} \ud k <\infty$, the annihilation operator is continuous from $D(L^{1/2})$ to $D(N^{-1/2})$  and one can interpret the expression~\eqref{eq:H formal} as the quadratic form
 \begin{equation}\label{eq:QF}
  \langle \psi, L \psi \rangle + \sum_{j=1}^M   \langle \psi , a(v(x_j-y)) \psi \rangle + \langle a(v(x_j-y))  \psi , \psi \rangle,
 \end{equation}
on $D(L^{1/2})\subset D(N^{1/2})$, since $a^*$ is the formal adjoint of $a$. When this form is bounded below, one defines the Hamiltonian $H$ to be the unique self-adjoint and semibounded operator associated with this form. This solves the problem of defining $H$, but yields only limited information, namely that $D(H)\subset D(L^{1/2})$ and that $H$ is semibounded.
\item When  $a(v(x-y))$ is not defined on $D(L^{1/2})$ one can still hope to construct $H$ using a renormalisation procedure due to Nelson~\cite{nelson1964}. In this procedure, one first regularises $v$, for example by replacing it by $v_\Lambda$ whose Fourier transform is $\hat v_\Lambda(k)=\hat v(k)\chi_\Lambda(k)$, where $\chi_\Lambda$ is the characteristic function of a ball of radius $\Lambda$. Then $v_\Lambda \in L^2$, so the operator $H_\Lambda$ with this interaction is self-adjoint on $D(L)$ for every $\Lambda\in \R_+$ and  $v_\Lambda$ converges to $v$ in $\mathscr{S}'(\R^d)$ as $\Lambda \to \infty$.
Under appropriate conditions on $v$ and $\omega$, one can then find (explicit) numbers $E_\Lambda$, so that
\begin{equation*}
 H_\infty=\lim_{\Lambda \to \infty}H_\Lambda + E_\Lambda
\end{equation*}
 exists in the norm resolvent sense and defines a self-adjoint and semibounded operator. This defines a Hamiltonian for the model up to a constant, since the numbers $E_\Lambda$ can always be modified by adding a finite constant in this procedure. However, one retains virtually no information on the domain of $H_\infty$, which led Nelson to pose in~\cite{nelson1964} the following problem: 
 \begin{quote}
It would be interesting to have a direct description of the operator $H_\infty$. Is $D(H_\infty) \cap D(L^{1/2}) = \{0\}$?
\end{quote} 
The second question was answered, affirmatively, in a recent article by Griesemer and Wünsch~\cite{GrWu17}. We will provide a direct description of $H_\infty$ and its domain in terms of abstract boundary conditions. From this description the answer to the second question will also be apparent.
\end{enumerate}

The models we consider will fall into one of these two classes. 
They are form perturbations of $L$, as under point (1) above, if $\int \frac{ |\hat v(k)|^2}{k^2+\omega(k)} \ud k <\infty$ and renormalisable in the sense of point (2) otherwise. The precise assumptions will be given in Condition~\ref{cond:alphabeta} below. 
The class of $v$ and $\omega$ we cover contains the following examples:
\begin{itemize}
 \item The Fröhlich model ($d=3$, $\omega=1$, $\hat v(k)=|k|^{-1}$) describes the interaction of nonrelativistic electrons with phonons in a crystal. As noted above, this model falls into the class of form perturbations. A recent exposition of the construction and an investigation of its domain can be found in the article of Griesemer and Wünsch~\cite{GrWu16}.
 \item The massive Nelson model ($d=3$, $\omega(k)=\sqrt{k^2+1}$, $\hat v(k)=\omega(k)^{-1/2}$) describes the interaction of nonrelativistic particles with relativistic, massive, scalar bosons, whose mass we have chosen to be one.
 It was defined rigorously by Nelson~\cite{nelson1964} and provides the blueprint for the renormalisation procedure described under point (2) above.
 \item Nonrelativistic point-particles in two dimensions ($d=2$, $\omega(k)=k^2+1$, $v=\delta$). In this model, the nonrelativistic ($x$-) particles interact with nonrelativistic bosons ($y$-particles) by creation/annihilation at contact. This is a  two-dimensional version of the model of~\cite{IBCpaper} with dynamical sources. The renormalisation procedure can be applied to this model by following Nelson's proof line-by-line (see also~\cite{GrWu17}).
\end{itemize}

\subsection{A Hamiltonian with abstract boundary conditions}

Our approach to constructing a Hamiltonian for these models starts not from the quadratic form or a regularisation of the expression~\eqref{eq:H formal}, but by considering extensions of $L$ to singular functions, adapted to the singularity of $v$.
This is analogous to the construction of Schrödinger operators with singular (pseudo-) potentials using the theory of self-adjoint extensions (see e.g.~\cite{albev, Be_etal2017, MiOt17,posi08}).
In those problems, one considers a self-adjoint operator $(S,D(S))$ (e.g. $S=-\Delta$ on $H^2(\R^d)$) and restricts it to the kernel of a singular \enquote{potential}. This could be the Sobolev trace on some lower dimensional set, the \enquote{boundary}, or some other linear functional on $D(S)$. The restriction of $S$ then defines a closed, symmetric operator $S_0$, and one searches for self-adjoint extensions of $S_0$, or, equivalently, restrictions of $S_0^*$. These extensions incorporate interactions through (generalised) boundary conditions.
We remark that, in many examples, such models can also be constructed using renormalisation techniques 
(see e.g.~\cite{DeFiTe1994, DiRa04, KiSi1995}), giving the same operators.
This is also true for our models, as we will show in Theorem~\ref{thm:renorm} below. 

Let $L_0$ be the restriction of $L$ to the domain
\begin{equation}\label{eq:D(L_0)}
 D(L_0)=D(L) \cap \ker \left( \sum_{j=1}^M  a(v(x_j-y)) \right).
\end{equation}
Then $L_0^*$ is an extension of $L$
whose domain contains, in particular, elements of the form
\begin{equation}
\psi=G\phi:=-g \left( \sum_{j=1}^M   a(v(x_j-y))L^{-1} \right)^*\phi
=-g L^{-1}  \sum_{j=1}^M   a^*(v(x_j-y)) \phi,
\end{equation}
for $\phi \in \hilb$. In this expression, $a^*(v(x_j-y))$ is to be understood as the adjoint of $a(v(x_j-y)):D(L)\to \hilb$ that maps $\hilb $ to $D(L)'=D(L^{-1})$, the dual of $D(L)$.
Note also that $L$ is invertible on the sectors with at least one boson since we assume $\omega\geq e_0> 0$. 
%

We will define an extension $A$ of $\sum_{j=1}^M  a(v(x_j-y))$ to functions in the range of $G$. One can then consider the operator $L_0^*+gA$ on the domain
\begin{equation*}
\left \{ \psi \in \hilb\vert \exists \phi \in \hilb: \psi -G\phi \in D(L)\right\}.
\end{equation*}
Since $G\phi \notin D(L)$ for $\phi \neq 0$, the function $\phi$ in this decomposition is unique.
The condition means that the singular part of $\psi\uppar{n}$, i.e.~the part not in $D(L)$, is determined by the \enquote{boundary value} $\phi\uppar{n-1}$. Note that, since $\mathscr{H}$ is the sum over all sectors $\mathscr{H}\uppar{n}$, the space on which the operator acts and the space of boundary values are both equal to $\hilb$.
The operator $L_0^* + gA$ is not symmetric on this domain, but it has symmetric  restrictions obtained by imposing boundary conditions, in the sense of linear relations between $\psi$ and $\phi$.

To find the boundary condition corresponding to the formal Hamiltonian~\eqref{eq:H formal}, first observe that the range of $G$ is contained in the kernel of $L_0^*$, because for all $\psi\in D(L_0)$
\begin{align}\label{eq: G Ker}
\langle L_0^* G \phi, \psi \rangle =  \langle  \phi,G^* L_0 \psi \rangle 
= - g \sum_{j=1}^M \langle    \phi, a(v(x_j-y)) \psi \rangle = 0 \, .
\end{align}
For any $\psi$ with $\psi -G\phi \in D(L)$ we then have
\begin{equation}\label{eq:L+a^*}
 L_0^* \psi= L_0^*(\psi - G\phi) = L(\psi -G\phi)=L\psi + g\sum_{j=1}^M a^*(v(x_j-y)) \phi.
\end{equation}
The final expression is a sum of vectors in $D(L)'$ that lies in $\hilb$, because it equals the left hand side. Imposing the relation $\phi=\psi$, i.e. that ~$\psi - G\psi \in D(L)$, then gives the equality
\begin{equation*}
 L_0^*\psi + gA\psi = L\psi + g\sum_{j=1}^M a^*(v(x_j-y)) \psi + gA \psi
\end{equation*}
in $D(L)'$. This is essentially the formal Hamiltonian~\eqref{eq:H formal}, but on a domain different from $D(L)$ chosen in such a way that the singularities of the first two terms cancel each other, and with the annihilation operator suitably extended to this domain. Our main result is that the Hamiltonian $H=L_0^*+gA$ is self-adjoint and bounded from below on the domain with this boundary condition. For the appropriate choice of extension $A$, it equals the Hamiltonian defined as a quadratic form, or by renormalisation, respectively.

Our hypothesis on $\hat v$ and $\omega$ is that they have upper, respectively lower, bounds by appropriate powers of $|k|$ or $1+k^2$, which is the case in all of the relevant examples. For simplicity we also set the rest-mass $e_0$ of the $y$-particles to one.
\begin{cond}\label{cond:alphabeta}
Let $v\in \mathscr{S}'(\R^d)$, $v\notin L^2(\R^d)$ and $\omega: \R^d \to \R_+$.
We have bounds $|\hat v(k) | \leq |k|^{-\alpha}$ and $\omega(k) \geq (1+k^2)^{\beta/2}$  with parameters $0\leq \alpha < \tfrac d2$,  $0\leq \beta \leq 2$ satisfying additionally one of the following two conditions:
\begin{enumerate}[label=(\arabic*)]
 \item $\alpha>\tfrac d2-1$ and thus $\int \frac{ |\hat v(k)|^2}{k^2+\omega(k)} \ud k <\infty$;
 \item $\int \frac{ |\hat v(k)|^2}{k^2+\omega(k)} \ud k =\infty$ and
 \begin{align*}
 \alpha=0 \text{ and } \beta>0 \text{ if } d&=2\\
 \alpha>\tfrac12 - \tfrac{\beta^2}{8+\beta^2} \text{ if } d&=3\,.
\end{align*} 
\end{enumerate}
\end{cond}
Note that the condition $\alpha< \tfrac d2$ implies $\hat v\in L^2_\mathrm{loc}$.
Later on, we will often state our results in terms of the parameter
 \begin{equation*}
  D:=d-2\alpha-2,
 \end{equation*}
which measures the (non)-integrability of $|\hat v(k)|^2(1+k^2)^{-1}$ and thus the singularity of the interaction. The first case of the condition corresponds to $D< 0$ and the second to $D\geq0$. 

\begin{definition}
Assume Condition~\ref{cond:alphabeta} holds and $d\in \{1,2,3\}$.
We define $A$ with domain $D(A)$ as the extension of 
\begin{equation*}
 \sum_{j=1}^M   a(v(x_j-y)):D(L)\to \hilb
\end{equation*}
given in
\begin{itemize}
 \item Equation~\eqref{eq:DnegAdef} if $\int \frac{ |\hat v(k)|^2}{k^2+\omega(k)} \ud k <\infty$, or
 \item Equations~\eqref{eq:DposAdef} and~\eqref{eq:DposTdef} if $\int \frac{ |\hat v(k)|^2}{k^2+\omega(k)} \ud k =\infty$.
\end{itemize}
\end{definition}

The integrability condition determines which of the cases in Condition~\ref{cond:alphabeta} applies.
Our main result is:

\begin{thm}\label{thm:main}
Let $d\in \{1,2,3\}$ and assume that $v$ and $\omega$ satisfy Condition~\ref{cond:alphabeta}. 
 Then the operator $H=L_0^*+gA$ with domain
 \begin{equation*}
  D(H)=\left \{ \psi \in \hilb\vert  \psi -G\psi \in D(L)\right\}
 \end{equation*}
 is self-adjoint and bounded from below. Its domain is contained in the domain of the number operator $N$ and for $\psi \in D(H)$ we have the equality 
 \begin{equation}\label{eq:H create}
  H\psi = L\psi + g\sum_{j=1}^M a^*(v(x_j-y))\psi + gA\psi.
 \end{equation}
in the dual of $D(L)$.
\end{thm}

For the Fröhlich model we are in the first case of Condition~\ref{cond:alphabeta} and have $\alpha=1$. For the Nelson model we can choose $\beta=1$, $\alpha=\tfrac12$. For $\beta=1$ the condition on $\alpha$ is $\alpha>\tfrac7{18}$, which also allows for slightly more singular cases. For our model of nonrelativistic point-particles in two dimensions the conditions are satisfied with $\beta=2$ and $\alpha=0$.
The corresponding model in one dimension, which is an extension of the one treated in~\cite{KeSi16} with moving sources, is a form perturbation. In fact, in one dimension we always have $\int \frac{ |\hat v(k)|^2}{k^2+\omega(k)} \ud k <\infty$ since we assume a bound with $0\leq\alpha <\tfrac 12$. 
For nonrelativistic bosons in three dimensions with $\beta=2$ our condition is $\alpha>\tfrac16$. This excludes $v=\delta$, corresponding to a model which is not known to be renormalisable (in sense of operators explained above).
However, our methods can be adapted to construct a Hamiltonian also in this case. This will be the subject of an upcoming publication by the first author~\cite{La18}.

Our result provides a self-adjoint operator $H$ whose action is given by~\eqref{eq:H formal}, if the separate terms are interpreted as elements of $D(L)'$ and $ \sum_{j=1}^M   a(v(x_j-y))$ is suitably extended. In the case of form perturbations, the annihilation operator is automatically well defined on $D(H)\subset D(L^{1/2})$. Our theorem then also implies that the quadratic form of $H$ is indeed given by the usual expression~\eqref{eq:QF}, since in this case Equation~\eqref{eq:L+a^*} also holds in the sense of quadratic forms on $D(L^{1/2})$.

 For the more singular models the extension of the annihilation operator involves an operation that can be interpreted as the addition of an \enquote{infinite constant}, and it is certainly not unique. 
 These models can also be treated by a renormalisation technique, see~\cite{GrWu17}. We make a choice of the extension $A$ for which $H$ coincides with the operator $H_\infty$ obtained by renormalisation (see also Remark~\ref{rem:Tdef}). The following theorem, proved in Section~\ref{sect:proof limit}, implies that $H=H_\infty$.

\begin{thm}\label{thm:renorm}
Let the conditions of Theorem~\ref{thm:main} be satisfied and $\int \frac{ |\hat v(k)|^2}{k^2+\omega(k)} \ud k =\infty$.
 For $\Lambda \in \R_+$ let $H_\Lambda$ be the Hamiltonian with the regularised interaction defined by $\hat v_\Lambda(k)=\hat v(k)\chi_\Lambda(k)$, where $\chi_\Lambda$ is the characteristic function of a ball of radius $\Lambda$, and let
 \begin{equation*}
  E_\Lambda= g^2 M \int \frac{ |\hat v_\Lambda(k)|^2}{k^2 + \omega(k)} \ud k\,.
 \end{equation*}
Then $H_\Lambda + E_\Lambda$ converges to $H$ in the strong resolvent sense.
\end{thm}

The domain of $H$ is explicit and for any given $\psi\in \hilb$ it is easy to check whether it belongs to $D(H)$ or not. In particular, the regularity properties of 
$\psi\in D(H)$ are easily deduced from the regularity of $G\psi$.
This allows us to answer Nelson's second question.

 \begin{cor}\label{cor:Nelson}
   Let the conditions of Theorem~\ref{thm:main} be satisfied and additionally $\omega \in L^\infty_{\mathrm{loc}}(\R^d)$. Then $D(H)\subset D(L^{1/2})$ if and only if $\int \frac{ |\hat v(k)|^2}{k^2+\omega(k)} \ud k <\infty$. Moreover, if $\int \frac{ |\hat v(k)|^2}{k^2+\omega(k)} \ud k =\infty$, then $D(H)\cap D(L^{1/2})=\{0\}$. 
 \end{cor}

This corollary follows from our more precise discussion of the regularity properties of $D(H)$ in Section~\ref{sect:regularity}. Essentially the same result for $M=1$ was recently obtained~\cite{GrWu16,GrWu17} by different methods.

The structure of the proof of our main result, Theorem~\ref{thm:main}, is essentially the same for the cases of form perturbations ($\int \frac{ |\hat v(k)|^2}{k^2+\omega(k)} \ud k <\infty$) and renormalisable models ($\int \frac{ |\hat v(k)|^2}{k^2+\omega(k)} \ud k =\infty$). However, the technical difficulties are slightly different, and much greater in the second case. For this reason, we will give the proof of the first case separately, in Section~\ref{sect:form}. This may also serve as a less technical presentation of the general strategy. The proof for the second case will be given in Section~\ref{sect:renorm}.
In both cases, the crucial technical ingredients of the proof are bounds on the operator $T=gAG$ that are sufficiently good regarding both regularity and particle number. This operator also appears in the theory of point interactions (with $v=\delta$), where it is known as the Ter-Martyrosyan--Skornyakov operator, see e.g.~\cite{Co_etal15,DeFiTe1994, MoSe17, MoSe18}.
We will build on some of the results obtained in this area, as we explain in Remark~\ref{rem:T}.

\section{Form perturbations}
\label{sect:form}

In this section we will prove Theorem~\ref{thm:main} under the assumptions of the first case in Condition~\ref{cond:alphabeta}. That is, we assume that $\omega(k) \geq 1$, $v\in \mathscr{S}'(\R^d)$, $v\notin L^2(\R^d)$,
and that $|\hat v(k) | \leq |k|^{-\alpha}$  for some $\tfrac d2 >\alpha > \tfrac d2 -1 $, respectively $d-2\alpha-2=D<0$.
We will use the notation
\begin{equation}\label{eq:aV}
 a(V):=\sum_{i=1}^{M} a(v(x_i-y)). 
\end{equation}
Under the assumptions of this section $a(V)$ is operator-bounded by $L$, as will be proved in Lemma~\ref{lem:Dnegbasicbound} below.

In the following we will often work in the Fourier representation. We denote by $P=(p_1,\dots, p_M)$, $K=(k_1,\dots, k_{n})$ the conjugate Fourier variables to $X=(x_1,\dots x_M)$, $Y=(y_1,\dots, y_{n})$.
The vector $\hat Q_j\in \R^{d 
\nu} $ is $Q\in \R^{d(\nu+1)}$ with the $j$-th entry deleted and $e_i$ is the inclusion of the $i$-th summand in $\R^{d\mu}=\bigoplus_{i=1}^\mu \R^{d}$. We will denote the Fourier representation of the operator  $L$ (on the $n$-boson sector) as multiplication by the function 
\begin{equation*}
L(P,K):=P^2+ \sum_{j=1}^n \omega(k_j)=:P^2 + \Omega(K). 
\end{equation*}

\begin{lem}
\label{lem:Dnegbasicbound}
If Condition~\ref{cond:alphabeta} holds with $D<0$ then
\begin{equation*}
 a(V)L^{-\frac12} N^{-\frac{2+D}{4}}
\end{equation*}
is a bounded operator on $\hilb$.
\begin{proof}
Since we are not concerned with the dependence of the norm on $M$ it is sufficient to estimate one term in the sum~\eqref{eq:aV} and then bound the norm of the sum by the sum of the norms.

In Fourier representation, we have
\begin{align*}
&\left(\Fou a(v(x_1-y)) L^{-\frac{1}{2}} N^{-\frac{2+D}{4}} \psi \right) \uppar {n}(P,\hat{K}_{n+1}) 
\\
&
= 
\sqrt{n+1} \int_{\R^d} \frac{\overline{\hat v(k_{n+1})}\hat \psi \uppar {n+1} (P-e_1k_{n+1},K)}{L(P-e_1k_{n+1},K)^{\frac12} (n+1)^{\frac{2+D}{4}}}     \, \ud k_{n+1} \, .
\end{align*}
To prove our claim it is sufficient to show that for some constant $C>0$ it holds that
\begin{align}
&
\abs{\int_{\R^d} \frac{\overline{\hat v(k_{n+1})}\hat \psi \uppar {n+1} (P-e_1k_{n+1},K)}{L(P-e_1k_{n+1},K)^{\frac12}}     \, \ud k_{n+1}}^2
\nonumber \\ \label{eq:Dnegbasicbound}
&
 \leq C  (n+1)^{\frac{D}{2}} \int_{\R^d}  \abs{\hat \psi \uppar {n+1} (P-e_1k_{n+1},K)}^2   \, \ud k_{n+1} \, ,
\end{align}
because we may afterwards integrate in $P$ and $\hat{K}_{n+1}$ and perform a change of variables $P \rightarrow P-{e}_1  k_{n+1}$. 


Using the Cauchy-Schwarz inequality, and our assumptions on $\hat v$ and $\omega$, we can bound the integral from above by

\begin{align*}
&\abs{\int_{\R^d} \frac{\overline{\hat v(k_{n+1})}\hat \psi \uppar {n+1} (P-e_1k_{n+1},K)}{L(P-e_1 k_{n+1},K)^{\frac{1}{2}}}     \, \ud k_{n+1}}^2
 \\
 &
\leq
\int_{\R^d} \frac{\abs{q}^{-2 \alpha}}{(p_1- q)^2 + n+1} \, \ud q \int_{\R^d} \abs{ \hat\psi \uppar {n+1}(P-{e}_1 k_{n+1},K)}^2 \, \ud k_{n+1} \, .
\end{align*}
The integral in $q$ takes its maximal value at $p_1=0$, by the Hardy-Littlewood inequality.
Rescaling by $(n+1)^{-1/2}$ then yields the upper bound
\begin{align*}
(n+1)^{-1+\frac{d-2\alpha}{2}} \int_{\R^d} \frac{\abs{q'}^{-2 \alpha}}{q'^2 +1} \, \ud  q' \int_{\R^d} \abs{ \hat\psi \uppar {n+1}(P-{e}_1 k_{n+1},K)}^2  \ud k_{n+1} \,,
\end{align*}
and this proves the claim.
\end{proof}
\end{lem}

\subsection{The extended domain}

Lemma~\ref{lem:Dnegbasicbound} has several important consequences. First of all, $a(V): D(L) \rightarrow \hilb$ is continuous in the graph norm of $L$. Thus $D(L_0)$, defined in~\eqref{eq:D(L_0)} as the kernel of $a(V)$ in $D(L)$, is a closed subspace of $D(L)$ with this norm. Due to our assumption that $v\notin L^2$, this subspace is also dense in $\hilb$.

\begin{lem}
If  Condition~\ref{cond:alphabeta} is satisfied
 the space $D(L_0)$ is dense in $\hilb$.
\end{lem}
\begin{proof}
The Hilbert space $\hilb$ is equal to the direct integral $\hilb=\int^\oplus_{\R^{Md}} \Gamma(L^2(\R^d)) \ud X$. We start by proving that for almost every $X\in \R^{Md}$ the kernel of $a(V(X))=\sum_{i=1}^M a(v(x_i-y))$ is dense in $\Gamma(L^2(\R^d))$.

The first step is to show that the kernel of  the linear functional defined by $\hat V(X,k)=\sum_{i=1}^M \ue^{\ui k x_i} \hat v(k)\in L^2_\mathrm{loc}(\R^d)$ is dense in $L^2(\R^d)$ if $\hat V(X)\notin L^2(\R^d)$.
The set of $X$ where $\hat V(X)\in L^2$ has  measure zero, see Lemma~\ref{lem:V} in the appendix.
Let $\hat v_\Lambda(k)=\hat v(k) \chi_\Lambda(k)$ with the characteristic function $\chi_\Lambda$ of a ball of radius $\Lambda>0$ and let $\hat V_\Lambda$ be defined like $\hat V$, with $\hat v_\Lambda$ replacing $\hat v$.
Let $f\in H^2(\R^d)$ and set 
\begin{equation*}
 \hat f_\Lambda(X,k) := \hat f(k) - \frac{\hat V(X,k)\chi_\Lambda(k)}{\int |\hat V_\Lambda(X,k')|^2 \ud k'}  \int \overline{\hat V(X,k')} \hat f(k') \ud k'.
\end{equation*}
Now $\lim_{\Lambda \to \infty} \int |\hat V_\Lambda(X,k')|^2 \ud k'=\infty$, because $\hat V(X,k) \notin L^2(\R^d)$, so $\hat f_\Lambda(k)$ converges to $\hat f(k)$ in $L^2(\R^d)$. On the other hand $\int \overline{\hat V(X,k)} \hat f_\Lambda(k) \ud k =0$ so, after taking the inverse Fourier transform in $k$, $f_\Lambda$ is in the kernel of $V(X)$. 

This implies that coherent states generated by functions in $H^2(\R^d) \cap \ker(V(X))$ are dense in $\Gamma(L^2(\R^d))$, see e.g.~\cite[Prop.12]{IBCpaper}. 
Such states are in the kernel of $a(V(X))$ since for the coherent state $\Phi(f)$ generated by $f$ we have $a(V(X))\Phi(f)=\langle V(X),f \rangle \Phi(f)$. 
Consequently, the kernel of $a(V(X))$ is dense in $\Gamma(L^2(\R^d))$ for almost every $X$.

To conclude the proof, notice that the approximants $f_\Lambda(X)$ above are in $H^2(\R^d)$ and depend smoothly on $X$. We can thus approximate any $\Gamma(L^2(\R^d))$-valued $L^2$-function of $X$ by smooth functions taking values in the kernel of $a(V(X))$.
Such functions are elements of $D(L)$ and this proves the claim.
\end{proof}

We have established that $L_0$, the restriction of $L$ to the kernel of $a(V)$, is a densely defined, closed, symmetric operator. As explained in the introduction, we are going to extend $L$ to a subspace of the domain of $L_0^*$. This space is spanned by functions of the form $\psi+G\phi$ with $\psi\in D(L)$ and $\phi\in \hilb$, where
\begin{align}
\label{eq:defofG}
G \phi 
&
:= \left(-g a(V) L^{-1}\right)^* \phi = - g L^{-1} a^*(V) \phi.
\end{align}
The operator $G$ is bounded on $\hilb$ by Lemma~\ref{lem:Dnegbasicbound}. It maps $\hilb$ to the kernel of $L_0^*$ by Equation~\eqref{eq: G Ker}.
%
%
Application of $G$ also improves regularity or decay in the particle number.
\begin{lem}\label{lem:DnegGanda}
If Condition~\ref{cond:alphabeta} holds with $D<0$ the operator $G$ is continuous from $\hilb$ to $D(N^{-D/4})$ and from $D(N^{\frac{2+D}4})$ to $D(L^{1/2})$.
\end{lem}
\begin{proof}
 In view of Equation~\eqref{eq:defofG} and the fact that $L\geq N$, this is immediate from Lemma~\ref{lem:Dnegbasicbound}.
\end{proof}

The next lemma is concerned with the map $1-G$ which is not only bounded but also invertible.
\begin{lem}
\label{lem:Dnegnisleftinv}
Assume Condition~\ref{cond:alphabeta} holds with $D<0$. Then $1-G$ is invertible and there exists a constant $C>0$ such that
\begin{align}
\label{eq:Dnegnisleftinv}
\norm{N \psi }_\hilb \leq C (\norm{N (1-G) \psi }_\hilb + \norm{ \psi }_\hilb ) \, .
\end{align}
\begin{proof}
Due to Lemma~\ref{lem:DnegGanda} there is a constant $C>0$ such that sector-wise 
\begin{equation*}
 \norm{G}^2_{\hilb \uppar {n-1} \rightarrow \hilb \uppar n} \leq C n^{\frac{D}{2}}\,.
\end{equation*}
Using this we estimate the $k$-th power of $G$ acting on $\psi \in \hilb$ by
\begin{align*}
\norm{ G^k \psi}^2_\hilb 
&
=
 \sum_{n \geq 1} \norm{\left( G^k \psi\right)\uppar n}^2_{\hilb \uppar n}\\
&\leq
 \sum_{n \geq k} \prod_{\ell=1}^k \norm{G}^2_{\hilb \uppar {n-\ell} \rightarrow \hilb \uppar {n-\ell+1}} \norm{ \psi \uppar {n- k}}^2_{\hilb \uppar {n- k}} 
 \\
 &
\leq
 C \norm{ \psi }^2_{\hilb } \sup_{m\geq 0} \prod_{r=1}^{k}  (r+m)^{\frac{D}{2}}
\leq
 C \norm{ \psi }^2_{\hilb } (k!)^{\frac{D}{2}} \, .
\end{align*}
This implies that the Neumann series $\sum_{k\geq 0} G^k$ converges in $\hilb$, hence $1-G$ is invertible. 

To prove \eqref{eq:Dnegnisleftinv}, first note that $G$ is a bounded operator from $D(N)$ to itself, because it maps $\hilb\uppar{n}$ to $\hilb \uppar{n+1}$. Define for any $\mu \geq 0 $ a modified map by
\begin{align*}
G_\mu:=-g\left(a(V)(L+\mu^2 )^{-1}\right)^*
\end{align*}
The norm $\norm{ G_\mu}_{D(N)\to D(N)}:=c_\mu$ is decreasing in $\mu$, so for sufficiently large $\mu$ we have $c_\mu<1$. Then $(1-G_\mu)^{-1}$ is a bounded operator on $D(N)$ with norm at most $(1-c_\mu)^{-1}$. 
 By the resolvent formula we then have
 \begin{align*}
  \norm{N\psi}&\leq (1-c_\mu)^{-1}\left(\norm{N(1-G_\mu)\psi}+\norm{(1-G_\mu)\psi}\right)\\
  &\leq (1-c_\mu)^{-1}\left(\norm{N(1-G)\psi}+\norm{\mu^2 N (L+\mu^2)^{-1}G_\mu\psi}+\norm{(1-G_\mu)\psi}\right).
 \end{align*}
Since $L\geq N$ this proves the claim. 
\end{proof}
\end{lem}

\subsection{The annihilation operator $A$}

So far we have considered $a(V)$ as an operator on $D(L)$. In view of Lemma~\ref{lem:Dnegbasicbound} we may also define it sector-wise on $D(L^{1/2})$ in the case $D<0$ of the current section. By Lemma~\ref{lem:DnegGanda} the annihilation operator thus makes sense on $G\phi\uppar{n}$, for any $n\in \N$.

\begin{lem}\label{lem:DnegTrelbounded}
 Assume that Condition~\ref{cond:alphabeta} holds with $D<0$ and 
 let $T=ga(V)G$. Then $T$ defines a symmetric operator on the domain $D(T)=D(N^{1+D/2})$.
\end{lem}
\begin{proof}
 Using Equation~\eqref{eq:defofG} we can write $T$ as
 \begin{equation*}
 T=- G^*LG= -g^2 \left(a(V) L^{-\frac12}\right)\left( a(V)L^{-\frac12}\right)^* .
 \end{equation*}
 This defines a continuous operator from $D(N^{1+D/2})$ to $\hilb$ by Lemma~\ref{lem:Dnegbasicbound}, and this operator is clearly symmetric.
 \end{proof}

On the set $D(A)=D(L)\oplus G D(T)$, which contains $D(H)$, we now define the annihilation operator $A$ by
\begin{align}
gA(\psi+G\phi):=g a(V)(\psi + G\phi) = ga(V)\psi + T\phi.
\label{eq:DnegAdef}
\end{align}

\begin{rem}
 The objects we have discussed so far occur naturally in the context of abstract boundary conditions.
 Let $K$ denote the restriction of $L_0^*$ to $D(A)=D(L)\oplus G D(T)$ and denote by $B(\eta +G\phi)=\phi$ a left inverse of $G$.
 Then $(D(T), B,-gA)$ is a quasi boundary triple for $K$ in the sense of Behrndt et al.~\cite{Be_etal2017}. In particular we have the identity
 \begin{equation*}
  \langle K \phi, \psi \rangle - \langle  \phi, K \psi \rangle = -\langle g A\phi, B\psi \rangle + \langle B\phi, gA \psi \rangle.
 \end{equation*}
The family of operators $G(z)=-g(L+z)^{-1}a^*(V)$ are called the $\gamma$-field, and $T(z)=gAG(z)$ the Weyl-operators associated to this triple.

In specific cases the operators $B$ and $A$ can be expressed as local boundary value operators on the configurations where at least one $x$-particle (source) and one $y$-particle (boson) meet, see~\cite{TeTu15, IBCpaper} and also Remark~\ref{rem:Tdef} for details.  
\end{rem}

\subsection{Proof of Theorem~\ref{thm:main} for $D<0$}\label{sect:proof form}

We will now prove that $H=L_0^* +gA$ is self-adjoint on the domain
\begin{equation*}
 D(H) = \lbrace \psi \in \hilb \vert (1 - G) \psi \in D(L) \rbrace=(1-G)^{-1} D(L)
\end{equation*}
in the case of form perturbations, $D<0$. The domain $D(H)$ is contained in $D(N)$ because $D(L)\subset D(N)$ and the domain of $N$ is preserved by $(1-G)^{-1}$, see Lemma~\ref{lem:Dnegnisleftinv}. 
We start the proof of self-adjointness by rewriting $H$ in a more symmetric form. First, we use the fact that $L_0^*G=0$, by Equation~\eqref{eq: G Ker}, to write for $\psi \in D(H)$
\begin{align}
H \psi &=
L_0^*(1-G)\psi  + gA \psi
\nonumber \\
&=L (1-G)\psi +  g a(V) (1-G)\psi+T\psi\,.
\label{eq:H rewrite v1}
\end{align}
Here, we have also used the \enquote{boundary condition} that $(1-G)\psi \in D(L)$ for $\psi \in D(H)$.
Since $G^*L=-g a(V)$ we can further rewrite this as
\begin{align}
H\psi & =(1-G^*) L (1-G)\psi + G^*  L (1-G)\psi +  g a(V) (1-G)\psi+T\psi\nonumber \\
%
%
&=(1-G)^* L (1-G)\psi +T\psi \label{eq:H rewrite v2} \, .
\end{align} 
We will prove that $H$ is self-adjoint by showing that it is a perturbation of the self-adjoint operator $ (1-G)^* L (1-G)$.
\begin{lem}
\label{lem:dressed}
The operator $H_0:=(1-G)^* L (1-G)$ is self-adjoint on $D(H)$ and positive.
\end{lem}
\begin{proof}
The operator $H_0$ is clearly positive and symmetric on $D(H_0)=D(H)$, so it suffices to show that $D(H_0^*)\subset D(H)$.
If $\phi \in D(H_0^*)$, $\psi\in D(H_0)=(1-G)^{-1}D(L)$, we have
\begin{equation*}
 \langle \phi, H_0 \psi \rangle = \langle (1-G)\phi, L (1-G)\psi \rangle,
\end{equation*}
and thus $(1-G)\phi\in D(L^*)=D(L)$. This proves the claim.
\end{proof}

To prove self-adjointness of $H$ we now show that $T$ is infinitesimally $H_0$-bounded.
By Lemma~\ref{lem:DnegTrelbounded} and Young's inequality we have, keeping in mind that $D<0$,
\begin{align}
\label{eq:relboundifpowerbounded}
&\norm{T \psi }_\hilb \leq C \norm{N^{1+D/2} \psi }_\hilb \leq \frac{C}{2} \left( (2+{D})  \eps \norm{N \psi}_\hilb -  D \eps^{\frac{2+D}{D}} \norm {\psi }_\hilb \right)  \, ,
\end{align}
for any $\eps>0$. Now Lemma~\ref{lem:Dnegnisleftinv} together with $L\geq N$ yields the inequality
\begin{align}
\label{eq:NisrelboundedbyH0}
 \norm{N \psi }_\hilb 
 &
 \leq C (\norm{N (1-G) \psi }_\hilb + \norm{ \psi }_\hilb )
 \nonumber \\
 &
 \leq C\left(\norm{(1-G)^{-1}} \norm{(1-G)^* L (1-G)}+\norm{\psi}\right).
\end{align}
This proves an infinitesimal bound on $T$ relative to $H_0$ and thus that $H=H_0+T$ is self-adjoint on $D(H)$, by the Kato-Rellich theorem.

\section{Renormalisable models}
\label{sect:renorm}

In this section we will deal with models falling into the second case of Condition~\ref{cond:alphabeta}.
This means that $\abs{\hat {v}(k)}\leq \abs{k}^{-\alpha}$ for some $\alpha \geq 0$, $\int \abs{\hat {v}(k)}^2(k^2+\omega(k))^{-1} \ud k=\infty$ (so necessarily $2\alpha\leq d-2$) and $\omega(k)\geq (1+k^2)^{\beta/2}$ for some $0 < \beta \leq 2$. In dimension $d=2$ this leaves  $\alpha=0$ as the only case. In $d=3$ we assume 
\begin{equation*}
\frac12 \geq \alpha>\frac12 - \frac{\beta^2}{\beta^2+8}.
\end{equation*}
In terms of $D=d-2\alpha-2$ this means that 
\begin{equation}\label{eq:Dbeta}
 0\leq D < \frac{2\beta^2}{\beta^2+8} \leq \frac\beta 2.
\end{equation}

Following the structure of Section~\ref{sect:form},
we start this section by discussing the extended domain. We then turn to the definition of the annihilation operator $A$ and finally prove Theorem~\ref{thm:main} and Theorem~\ref{thm:renorm}.

\subsection{The extended domain}

As in Section~\ref{sect:form}, we consider the extension of $L$ (or the restriction of $L_0^*$) to vectors of the form $\psi + G\phi$ with $\psi\in D(L)$, $\phi\in \hilb$ and $G=-gL^{-1}a^*(V)$. We start by discussing the mapping properties of $G$,
showing in particular that $a^*(V):\hilb \to D(L)'$ and $a(V):D(L)\to \hilb$ are continuous.
In Section~\ref{sect:form}, where $D<0$, we showed that $G$  maps into the form domain of $L$. For $D\geq0$ however, $G$ will not map into the form domain of $L$ but instead into $D(L^\eta)$ for some $\eta < \frac{2-D}{4} \leq \frac{1}{2} $. We first prove a bound on $G$ that will allow us to use the regularity and the decay in the particle number $N$ in an optimal way later on.

\begin{prop}
\label{prop:generalboundong}
Let Condition~\ref{cond:alphabeta} be satisfied and define the affine transformation $u(s):=\frac{\beta}{2} s - \frac{D}{2}$.
Then for any $s\geq0$ such that $u(s)<1$ and all $0\leq \eta < \frac{1+u(s)-s}2$ 
there exists a constant $C$ such that for all $n\in \N$
\begin{equation*}
 \norm{L^\eta G \psi}_{\hilb\uppar{n+1}} \leq C\left(1+n^{\frac{\max(0,1-s)}2}\right) \norm{\psi\uppar{n}}_{\hilb\uppar{n}}\,.
\end{equation*}
\end{prop}
\begin{proof}
 Note first that, since $\beta\leq 2$, the function $u(s)-s$ is non-increasing and thus $\eta < \frac{1+u(0)}{2}= \frac{2-D}{4} $. The expression for the Fourier transform of $G\psi\uppar{n}$ is given by
 \begin{equation}\label{eq:G Fourier}
  \widehat{G\psi\uppar{n}}(P,K)= \frac{-g}{\sqrt{n+1}}\sum_{i=1}^M  \sum_{j=1}^{n+1} \frac{\hat v(k_j) \hat \psi \uppar n(P+{e}_i k_j,\hat{K}_j)}{L(P,K)}.
 \end{equation}
 As we are not interested in the dependence of the constant $C$ on $M$ or $g$ it is sufficient to estimate  the $\hilb\uppar{n+1}$-norm of the expression
\begin{align}
\gamma  \hat\psi \uppar n(P,K)
= \frac{1}{\sqrt{n+1}} \sum_{j=1}^{n+1} \frac{\hat v(k_j) \hat \psi \uppar n(P+{e}_1 k_j,\hat{K}_j)}{L(P,K)^{1-\eta}}\,.
\end{align}
We first multiply by $\omega(k_j)^{\frac{s}{2}}$ and its inverse, and then use the finite-dimensional Cauchy-Schwarz inequality to obtain
\begin{align}
\label{eq:gcauchyschwarz}
\abs{\gamma  \hat\psi \uppar n(P,K)}^2 
&
\leq \frac{1}{n+1} \sum_{i,j=1}^{n+1}  \frac{\abs{\hat v(k_j)}^2 \abs{\hat \psi \uppar n(P+{e}_1 k_j,\hat{K}_j)}^2}{L(P,K)^{2(1-\eta)} \omega(k_j)^s}\omega(k_i)^s
 \, .
\end{align}
Let $\abs{\gamma \uppar d \hat \psi \uppar n}^2$ denote the sum of terms in this sum with $i=j$, and $\abs{\gamma \uppar {od} \hat \psi \uppar n}^2$ the sum of the remaining terms.
We have
\begin{align}
\label{eq:Dposg1}
\abs{\gamma \uppar d \hat \psi \uppar n}^2
&
= 
\frac{1}{n+1} \sum_{j=1}^{n+1}  \frac{\abs{\hat v(k_j)}^2 \abs{\hat \psi \uppar n(P+{e}_1 k_j,\hat{K}_j)}^2}{L(P,K)^{2(1-\eta)} }\,,
\\ 
\label{eq:Dposg2}
\abs{\gamma \uppar {od} \hat \psi \uppar n}^2
&
\leq \frac{ n^{\max(0,1-s)}}{n+1} \sum_{j=1}^{n+1} \frac{\abs{\hat v(k_j)}^2 \abs{\hat \psi \uppar n(P+{e}_1 k_j,\hat{K}_j)}^2}{L(P,K)^{2(1-\eta)} \omega(k_j)^s}  \Omega(\hat{K}_j)^s \, .
\end{align}
In the second line we have used the bound (with the notation $\sum_{j \in J} \omega(q_j) = \Omega(Q)$)
\begin{align}
\label{eq:sum}
\sum_{i=1}^n \omega(k_i)^s \leq {n}^{\max(0,1-s)} \Omega(K)^{s} \, ,
\end{align} 
which for $s<1$ follows from the Hölder inequality, while for $s\geq 1$ it holds by interpolation between the $\ell^1$-norm and the $\ell^\infty$-norm.

Note that both sums in~\eqref{eq:Dposg1},~\eqref{eq:Dposg2} are just symmetrisations and every summand has the same integral. 
%
Integrating~\eqref{eq:Dposg1} and performing a change of variables thus yields
\begin{align*}
\int & \abs{\gamma \uppar d \hat \psi \uppar n(P,K) }^2\, \ud P \ud K 
=
\int \frac{\abs{\hat v(k_{n+1})}^2 \abs{\hat \psi \uppar n(P,\hat{K}_{n+1})}^2}{L(P-{e}_1 k_{n+1},K)^{2(1-\eta)} }  \, \ud P \ud K \, .
\end{align*}
We notice that the square of $\hat\psi \uppar n$ does not depend on $k_{n+1}$ anymore. 
Using that $4 \eta <2-D=4+2\alpha-d$, and the Hardy-Littlewood inequality, the integral over $k_{n+1}$ can be bounded by
\begin{align}\notag
 \int_{\R^d}  &\frac{\abs{\hat v(k_{n+1})}^2 }{L(P-{e}_1 k_{n+1},K)^{2(1-\eta)} }  \, \ud k_{n+1} \\
&\leq  
\int_{\R^d}  \frac{\abs{k_{n+1}}^{-2 \alpha}}{((p_1-k_{n+1})^2+\Omega(\hat{K}_{n+1})+1)^{2(1-\eta)} }  \,  \ud k_{n+1} 
\notag\\
&
\leq
C   (\Omega(\hat{K}_{n+1})+1)^{-2(1-\eta)-\alpha+\frac{d}{2} }  \, ,
\label{eq:G int bound}
\end{align}
The exponent here is negative, which proves the required bound for $\abs{\gamma\uppar{d}\psi\uppar{n}}^2$.

The integration of~\eqref{eq:Dposg2} gives 
\begin{align*}
\int & \abs{\gamma \uppar {od} \hat \psi \uppar n(P,K)}^2 \, \ud P \ud K \\
&\leq n^{\max(0,1-s)}\int \frac{\abs{\hat v(k_{n+1})}^2 \Omega(\hat{K}_{n+1})^{s} \abs{ \hat \psi \uppar n(P,\hat{K}_{n+1})}^2}{L(P-{e}_1 k_{n+1},K)^{2(1-\eta)} \omega(k_{n+1})^s }  \, \ud P \ud K \, .
\end{align*}
The condition $u(s)<1$ implies that $\beta s +2\alpha <d$, so we can bound the integral in $k_{n+1}$ by
\begin{align*}
\int_{\R^d}   \frac{\abs{k_{n+1}}^{-2 \alpha-\beta s}}{((p_1-k_{n+1})^2+\Omega(\hat{K}_{n+1}))^{2(1-\eta)} }  \,  \ud k_{n+1} 
\leq
C   \Omega(\hat{K}_{n+1})^{-2(1-\eta)-\frac{2\alpha+\beta s}{2}+\frac{d}{2} }  \, .
\end{align*}
It follows that
\begin{align*}
 \int &  \abs{\gamma \uppar {od} \hat \psi \uppar n(P,K)}^2 \, \ud P \ud K\\
 &\leq C n^{\max(0,1-s)}\int  \Omega(\hat{K}_{n+1})^{2\eta -1 - u(s) +s} \abs{\hat \psi \uppar n(P,\hat{K}_{n+1})}^2\, \ud P \ud \hat K_{n+1} \, .
\end{align*}
The exponent  of $\Omega(\hat K_{n+1})$ in this integral is negative by hypothesis, and this proves the claim.
\end{proof}

A simple consequence of this proposition is that $G$ maps $\hilb$ into the domain of some power of $L$, and thus also of $N$.

\begin{cor}
\label{lem:gonfockspace}
Assume  Condition~\ref{cond:alphabeta} holds with $D\geq0$. There exists  an $\eta \in (0,1/2)$ such that $G$ is a continuous operator from $\hilb$ to $D(L^\eta)$.
\end{cor}
\begin{proof}
 We apply Proposition~\ref{prop:generalboundong}, distinguishing two cases. First, if $D=0$ and $\beta=2$, then $u(s)=s$ and we choose, for some $1>\eps>0$, $s_\eps=1-\eps$ and $\eta_\eps=\frac{1-\eps}2$.
 Proposition~\ref{prop:generalboundong} then gives the bound 
 \begin{equation*}
\norm{ L^{\frac{1-\eps}2}G\psi}_{\hilb \uppar{n+1}} \leq C (1+n^{\eps/2}) \norm{\psi\uppar n}_{\hilb \uppar n}. 
 \end{equation*}
 Dividing  by $(1+n^{\eps/2}) \leq c L^{\eps/2}$ then shows that $G$ maps $\hilb $ to $D(L^{1/2-\eps})$  for all $0<\eps\leq \frac12$, in this case.
 
 In all other cases, we have $u(1)=(\beta-D)/2<1$, by~\eqref{eq:Dbeta}, and we may choose in Proposition~\ref{prop:generalboundong} $s=1$ and any $0\leq\eta<\frac{\beta-D}4$.
\end{proof}

An important consequence of this is that $ga(V)L^{-1}=-G^*$ is a continuous operator on $\hilb$, so $a(V)$ is well defined on $D(L)$. 
We can thus define $L_0$ and its adjoint in the very same way as in Section~\ref{sect:form}. We can also prove the analogue of Lemma~\ref{lem:Dnegnisleftinv}.
\begin{lem}
\label{lem:Dposnisleftinv}
Let Condition~\ref{cond:alphabeta} be satisfied. Then $1-G$ is invertible and there exists a constant $C>0$ such that
\begin{align}
\label{eq:Dposnisleftinv}
\norm{N \psi }_\hilb \leq C (\norm{N (1-G) \psi }_\hilb + \norm{ \psi }_\hilb ) \, .
\end{align}
\end{lem}
\begin{proof}
Using Corollary~\ref{lem:gonfockspace} and the fact that $N\leq L$ the proof for the case $D\geq0$ is exactly the same as in Lemma~\ref{lem:Dnegnisleftinv} for $D<0$.
\end{proof}

\subsection{Extending the annihilation operator for $D\geq0$}
In this section we will extend the annihilation operator $a(V)$ to certain vectors in the range of $G$, defining the operator $A$. To do so, for any symmetric operator $(T,D(T))$ we could define an extension $g A$ on the set $D(A)=D(L) \oplus G D(T)$ by
\begin{align}\label{eq:DposAdef}
g A(\psi + G \phi) := g a(V) \psi +g A G \phi =  g a(V) \psi +T \phi\,.
\end{align}
In the case of a form perturbation, where $G$ maps sector-wise into $D(L^{1/2})$, the right extension of $a(V)$ to these elements is obviously $a(V)$ itself. As a result, we have simply chosen $T = g a(V) G$ in Section~\ref{sect:form}.
However, this choice is not possible if the domain of $a(V)$ and the range of $G$ do not match, as is the case if $D\geq 0$. We will define $T$ by slightly modifying the expression for $ga(V)G$, in such a way that the operator $H$ we obtain coincides with the one constructed by renormalisation.
In Fourier representation, $ga(V)G$ is formally given by
\begin{align*}
g \sqrt{n+1}  \sum_{\ell=1}^M \int_{\R^d} \overline{\hat v(k_{n+1})} \widehat {G  \phi\uppar n}(P-{e}_\ell k_{n+1},K) \, \ud k_{n+1} \, .
\end{align*}
Expanding the formal action by spelling out $\widehat {G  \phi\uppar n}$ as in~\eqref{eq:G Fourier} gives
\begin{align}
\label{eq:Tformally v1} 
-{g^2} \sum_{\ell=1}^M \sum_{i=1}^M \sum_{j=1}^{n+1} \int_{\R^d}  \frac{\overline{\hat v(k_{n+1})} \hat v(k_{j}) \hat \phi\uppar n(P-{e}_\ell k_{n+1}+{e}_i k_j,\hat{K}_j)}{L(P-{e}_\ell k_{n+1},K) } \, \ud k_{n+1} \, .
\end{align}
Have a look at the sum above. In the terms where $j=n+1$ and $i=\ell$, the function $\hat \phi \uppar n$ does not depend on the variable $k_{n+1}$ anymore. Formally, these terms define a multiplication operator, with the multiplier  given by a sum over integrals of the form
\begin{align*}
-g^2 \int_{\R^d}  \frac{\abs{\hat v(k_{n+1})}^2  }{L(P-{e}_\ell k_{n+1},K) } \, \ud k_{n+1} \, .
\end{align*}
This is what we will call the \textit{diagonal} part in the following. However, this integral is divergent.
In order to obtain a well-defined operator, we replace this integral by a regularised version. We set
\begin{align}
\label{eq:Dposregintergal}
I_\ell(P,\hat{K}_{n+1}) 
&
:= 
\int_{\R^d} \abs{\hat v(k_{n+1})}^2  \left(\frac{ 1 }{L(P-{e}_\ell k_{n+1},K) } - \frac{ 1 }{k_{n+1}^2+\omega(k_{n+1}) } \right) \, \ud k_{n+1}
\end{align}
and define $T_\ud$, the diagonal part of $T$, in Fourier  representation as 
\begin{align} \label{eq:DposDefofTd}
&\widehat{T_\ud  \phi\uppar n}(P,\hat{K}_{n+1})
 := -g^2 \hat \phi\uppar n(P,\hat{K}_{n+1}) \sum_{\ell=1}^M  I_\ell(P,\hat{K}_{n+1}) \, .
\end{align}
The remaining expressions in~\eqref{eq:Tformally v1} constitute the \textit{off-diagonal} part of $T$. 
It is a sum of integral operators and we will show that they are defined on suitable spaces, without modification. Spelled out, we have
\begin{align}
\label{eq:Toffdiagdef v1}
\widehat{T_\uod \phi \uppar n}& (P,\hat K_{n+1})\\
:=
&
\begin{aligned}[t]
&
  -
{g^2}  \sum_{\ell=1}^M \sum_{\substack{i=1\\i\neq \ell}}^M \int_{\R^d}  \frac{\abs{\hat v(k_{n+1})}^2  \hat \phi\uppar n(P-({e}_\ell -{e}_i) k_{n+1},\hat{K}_{n+1})}{L(P-{e}_\ell k_{n+1},K) } \, \ud k_{n+1} 
\\
&
  -{g^2}  \sum_{\ell=1}^M \sum_{i=1}^M \sum_{j=1}^{n} \int_{\R^d}  \frac{\overline{\hat v(k_{n+1})} \hat v(k_{j}) \hat \phi\uppar n(P-{e}_\ell k_{n+1}+{e}_i k_j,\hat{K}_j)}{L(P-{e}_\ell k_{n+1},K) } \, \ud k_{n+1} \,.
\end{aligned}\notag
\end{align}
We define the operator 
\begin{align}
\label{eq:DposTdef}
T \phi \uppar n  := T_\ud \phi \uppar n  + T_\uod \phi \uppar n 
\end{align} 
by the expressions above, on a domain (or rather a family of admissible domains) to be specified in Proposition~\ref{prop:Tmainprop} below.
\begin{rem}
\label{rem:Tdef}
As noted before, the choice of the operator $T$ is not unique. In fact, any operator $T$ that is symmetric on an appropriate domain will lead to a self-adjoint operator $H$. 
We have made the choice for which this operator coincides with the one constructed by renormalisation, with the usual choice of renormalisation constant $E_\Lambda$, cf.~Theorem~\ref{thm:renorm}.
Observe that the the regularised integral~\eqref{eq:Dposregintergal} is formally obtained by subtracting the \enquote{constant} $E_\infty=\int \abs{\hat v(k)}^2(k^2+\omega(k))^{-1} \ud k$. In this sense, the operator $A$ may be viewed as the \enquote{renormalised} annihilation operator.

Another way to interpret the expression for $T_\ud$ is that the distribution $v(x_\ell-y_{n+1})$ is not applied to the function $G\phi\uppar{n}$, but to the more regular function 
\begin{equation*}
 G\phi\uppar{n} +g \phi\uppar{n}(X,\hat Y_{n+1}) f(x_\ell-y_{n+1})\,,
\end{equation*}
 where $\hat f(k)=\hat v(k) (k^2 + \omega(k))^{-1}$. Here, the second term effectively cancels the local singularities of $G\phi\uppar{n}$ in the directions parametrised by $x_\ell-y_{n+1}$.
 This point of view is particularly natural if $v(y)$ is singular only at $y=0$, and thus $G\phi\uppar{n}$ is singular on the planes $\{x_\ell=y_j\}$. In this case, the off-diagonal operator $T_\mathrm{od}$ comes from the application of $v(x_\ell-y_{n+1})$ to functions $L^{-1} v(x_i - y_j) \phi\uppar{n}(X,\hat Y_j)$ in directions where they are regular.

In concrete examples, there might be other criteria that single out a choice of $T$, respectively $A$. For example in the case of $v=\delta$, $d=2$, $\omega(k)=k^2+1$, the annihilation operator $a(V)$ is (the sum of) evaluation operators on the planes where $x_\ell=y_j$. These are local boundary values and one would want the extension $A$ to be local in this sense as well. In this example, the functions in the range of $G$ are singular, with an asymptotic expansion
\begin{equation*}
 G\phi\uppar{n}(X,Y)= \frac{c \log|x_\ell-y_j| \phi\uppar{n}(X,\hat Y_j)}{\sqrt{n+1}} + F(X,Y) 
\end{equation*}
as $|x_\ell-y_j|\to 0$, where $c$ is  a universal constant and $F$ is a function that has a (suitable) limit almost-everywhere  on $\{x_\ell=y_j\}$. One can view $\phi$ as a local boundary value of this function, since
\begin{equation*}
 \phi(X,\hat Y_j) = \sqrt{n+1} \lim_{|x_\ell-y_j|\to 0} \frac{G\phi\uppar{n}(X,Y)}{c \log|x_\ell-y_j|}.
\end{equation*}
It is then natural to choose $AG\phi\uppar{n}$ as the evaluation of the regular part $F(X,Y)$ of $G\phi\uppar{n}$, more precisely 
\begin{equation*}
  AG \phi\uppar{n}(X,\hat Y_{n+1}) = 
  \lim_{r\to 0}\sum_{\ell=1}^M\int\limits_{|x_\ell-y_{n+1}|=r} \hspace{-21pt}\left( \sqrt{n+1}G \phi\uppar{n}(X,Y) - c \log(r) \phi\uppar{n}(X,\hat Y_{n+1})\right)\ud \omega.
\end{equation*}
This is clearly a local boundary value, and one can check that this coincides with our choice of $A$ up to the addition of a global constant. Such boundary values are discussed in~\cite{La18, IBCpaper, TeTu15} for a variety of models involving creation and annihilation of particles.
Boundary values for a two-dimensional model with point interactions were treated by Dell'Antonio, Figari, and Teta~\cite[Sec.5]{DeFiTe1994}.
\end{rem}

The next proposition states the important mapping properties of $T$. 
For our model of non-relativistic point-particles in two dimensions ($d=2$, $v=\delta$, $\omega(k)=k^2+1$), we show that $T_\ud$ is defined on $D(L^\eps)$ for any $\eps>0$ (in fact, it is a Fourier multiplier of logarithmic growth), and that $T_\uod$ is a bounded operator on $\hilb\uppar{n}$ whose norm grows at most like $n^\eps$.
 For the Nelson model ($d=3$, $\omega(k)=\sqrt{1+k^2}$, $\hat v(k)=\omega(k)^{-1/2}$), $T_\ud$ is also bounded by any power of $L$, and  $T_\mathrm{od}$ is an operator $D(L^\eps) \cap \hilb\uppar{n} \to \hilb \uppar{n}$ whose norm grows at most like $n^{1-2\eps}$.

\begin{prop}
\label{prop:Tmainprop}
Assume Condition~\ref{cond:alphabeta} holds with $D\geq0$, set $u(s) = \frac{\beta }{2} s - \frac{D}{2}$ and define $T$ for every $n\in \N$ by the expression~\eqref{eq:DposTdef}.
\begin{itemize}
\item If $D=0$ and $\beta=2$ then, for any $\eps>0$, $T$ defines a symmetric operator on the domain $D(T)=D(L^{\eps})$. 
\item If either $D>0$ or $\beta<2$ then, for all $s > 0$ such that the following two conditions are satisfied
\begin{align*}
u(s) &< 1
\\
0 &< u(u(s)) \, ,
\end{align*}
the operator $T$ is symmetric on $D(T)=D((N+1)^{\max(0,1-s)} L^{s-u(s)})$.
\end{itemize} 
\end{prop}
\begin{proof}
The proof will be split into three lemmas. 
In Lemma~\ref{lem:DposTdiag} we deal with the diagonal operator $T_\ud$. We will show that $T_\ud$ defines a symmetric operator on the domain $D(L^{\max(\eps,D/2)})$ for any $\eps>0$.
We further decompose the \textit{off-diagonal} part in~\eqref{eq:Toffdiagdef v1} as
\begin{align*}
\widehat{T_\uod \phi} \uppar n := \sum_{\ell=1}^M \sum_{i=1,i\neq \ell}^M \theta_{i\ell} \hat \phi \uppar n  + \sum_{\ell=1}^M \sum_{i=1}^M \tau_{i\ell} \hat \phi \uppar n 
\end{align*}
 with
\begin{align}
\label{eq:T1def}
\theta_{i\ell}  \hat{\phi} \uppar n (P,\hat{K}_{n+1})  := \int_{\R^d}  \frac{\abs{\hat v(k_{n+1})}^2 \hat \phi\uppar n(P+({e}_i-{e}_\ell) k_{n+1},\hat{K}_{n+1})}{L(P-{e}_\ell k_{n+1},K) } \, \ud k_{n+1} 
\end{align}
and
\begin{align}
\label{eq:T2def}
\tau_{i\ell} \hat{\phi} \uppar n (P,\hat{K}_{n+1})  := \sum_{j=1}^{n} \int_{\R^d}  \frac{\overline{ \hat v(k_{n+1})} \hat v(k_{j}) \hat \phi\uppar n(P-{e}_\ell k_{n+1}+{e}_i k_j,\hat{K}_j)}{L(P-{e}_\ell k_{n+1},K) } \, \ud k_{n+1} \, .
\end{align}
 In Lemma~\ref{lem:DposT1} the properties of the $\theta$-terms and in Lemma~\ref{lem:T2mainlem} those of the $\tau$-terms are described. Both of these lemmas rely on modifications of the Schur test, but the second one will be more difficult due to the additional sum  over $n$ terms in $\tau_{i\ell}$.

If $D=0$ and $\beta=2$, Lemma~\ref{lem:DposTdiag} shows that $T_\ud$ is defined on $D(L^{\eps})$ for any $\eps>0$. Regarding the terms $\theta_{i\ell}$, Lemma~\ref{lem:DposT1} shows that they are bounded and that their sum is symmetric. Now because $u(s)=s$, the conditions on the parameter $s$ in Lemma~\ref{lem:T2mainlem} reduce to $s \in (0,1)$. The lemma then states that the operators $\tau_{i \ell} $ are defined on $D(N^{1-s})$ and their sum is symmetric. Choosing $s_\eps = 1 -\eps$ and estimating $N \leq L$ yields the claim in this case.

If either $D>0$ or $\beta<2$, strictly, we have for sufficiently small $\eps>0$
\begin{align*}
s-u(s) &= \frac{1}{2}(2-\beta) s + \frac{D}{2} \geq \max\left(\eps, \frac{D}{2}\right) \, .
\end{align*}
This means that $  D((N+1)^{\max(0,1-s)} L^{s-u(s)}) \subset D(L^{\max(\eps, D/2)}) $ for such an $\eps$. Therefore, Lemmas \ref{lem:DposTdiag} -- \ref{lem:T2mainlem} together prove the claim.
\end{proof}
\begin{lem}
\label{lem:DposTdiag}
Assume Condition~\ref{cond:alphabeta} holds with $D\geq 0$. Then for any $\eps>0$ the expression $T_\ud$ given by $\eqref{eq:DposDefofTd}$ defines a symmetric operator on the domain $D(T_\ud)= D(L^{\max(\eps,D/2)})$. 
\end{lem}
\begin{proof}
The integral~\eqref{eq:Dposregintergal} defining $T_\ud $ is real, so $T_\ud$ is a real Fourier multiplier and it is sufficent to prove that it maps the domain $D(T_\ud)$ to $\hilb$.
Specifying as usual to $\ell=1$ we have to show that there exists a constant $C >0$ such that the inequality
\begin{align}
\label{eq:Tdiaggeneral}
I_1(P,\hat{K}_{n+1}) \leq C \left(L(P,\hat{K}_{n+1})^{\max(\eps,D/2)} + 1 \right)
\end{align}
holds pointwise on $\R^{M d} \times \R^{ n d}$. We will use that
\begin{align*}
&
I_1(P,\hat{K}_{n+1}) 
=
\int_{\R^d} \abs{\hat v(k_{n+1})}^2  \frac{2 p_1 \cdot k_{n+1} - p_1^2  -  \left( \hat{P}_1^2 + \Omega(\hat{K}_{n+1}) \right) }{L(P-{e}_1 k_{n+1},K) (k_{n+1}^2+\omega(k_{n+1})) }  \, \ud k_{n+1} 
\end{align*}
and distinguish between $d=2$ and $d=3$. 

If $d=2$ then necessarily $\alpha=0$ and $D=0$. We denote the integration variable by $q$ instead of $k_{n+1}$ and also write $p$ for $p_1$. The absolute value of the integral $I_1$ can, for any $\eps \in (0,1)$, be bounded by
\begin{align*}
&\int_{\R^2}   \frac{2 \abs {p } \abs q + p^2  +  \left( \hat{P}_1^2 + \Omega(\hat{K}_{n+1}) \right) }{\left((p-q)^2+ \hat{P}_1^2 + \Omega(\hat{K}_{n+1}) \right) (q^2+1)}   \ud q 
\\
&
\leq
\int_{\R^2}   \frac{2 \abs {p}   (q^2+1)^{\frac{1}{2}} + p^2  }{\left((p-q)^2+ 1 \right) (q^2+1)}   \ud q  
+ \int_{\R^2}   \frac{ \hat{P}_1^2 + \Omega(\hat{K}_{n+1})  }{\left((p-q)^2+ \hat{P}_1^2 + \Omega(\hat{K}_{n+1}) \right) \abs{q}^{2(1-\eps)}}   \ud q   .
\end{align*}
The second term is bounded by some constant times $ (\hat{P}_1^2 + \Omega(\hat{K}_{n+1}))^{\eps}$. For the first term we use Lemma~\ref{lem:parameterint} in the appendix, which yields
\begin{align*}
 \int_{\R^2}   \frac{2 \abs{ p}  (q^2+1)^{\frac{1}{2}} + p^2  }{\left((p-q)^2+ 1 \right) (q^2+1)}  \, \ud q
&
\leq
3 C  (\log(1+\abs{p})+1)
\leq
\tilde C (\abs{p}^{\eps}+1),
\end{align*}
for some $\tilde{C}>0$.

Now let $d=3$ and $D>0$. The absolute value of the integral $I_1$ is bounded by
\begin{align}
%
\int_{\R^3}  \frac{2 \abs{p}  \abs q + p^2 }{\left((p-q)^2+ 1 \right)  \abs{q}^{2+2 \alpha}}  \, \ud q  + \int_{\R^3}   \frac{ \hat{P}_1^2 + \Omega(\hat{K}_{n+1})  }{\left((p-q)^2+ \hat{P}_1^2 + \Omega(\hat{K}_{n+1}) \right)  \abs{q}^{2+2 \alpha}}  \, \ud q \, .\label{eq:Ibound 3d}
\end{align}
The integrals converge because $2+2 \alpha = d- D < d$ and $\alpha>0$. The second term is easily seen to be bounded by a constant times $ (\hat{P}_1^2 + \Omega(\hat{K}_{n+1}))^{\frac{D}{2}}$. For the first term we can use Lemma~\ref{lem:parameterint} in the appendix which gives
\begin{align*}
& \int_{\R^3}  \frac{2 \abs {p}  \abs q + p^2 }{\left(({p} -q)^2+ 1 \right)  \abs{q}^{2+2 \alpha}}  \, \ud q 
\leq
\frac{2 C \abs{{p} }}{\abs{{p} }^{2 \alpha}} + \frac{ C p^2}{\abs{{p} }^{ 1 + 2 \alpha}} 
\leq
{\tilde C \abs{p }^D}
\leq
{\tilde C L(P,\hat{K}_{n+1})^{\frac{D}{2}}} \, .
\end{align*}
If $D=0$, the function $\abs{q}^{-2-2\alpha}=\abs{q}^{-d}$ is not locally integrable. We thus use the estimate $q^2+1 \geq q^{2(1-\eps)}$, for any $\eps\in (0,1)$. This yields a bound on $\abs{ I_1}$ as in Equation~\eqref{eq:Ibound 3d}, but with $\abs{q}^{-2-2\alpha}$ replaced by $\abs{q}^{-d+2\eps}$.
Applying Lemma~\ref{lem:parameterint} then gives a bound on $\abs{ I_1}$ by some constant times $L(P,\hat{K}_{n+1})^{\eps}$.
\end{proof}

\begin{lem}
\label{lem:DposT1}
Assume Condition~\ref{cond:alphabeta} holds with $D\geq 0$. Then, for any $i,\ell\in \{1,\dots, M \}$ with $i \neq \ell $, the operator $\theta_{i\ell}$ defined by~\eqref{eq:T1def} is continuous from $D(L^{D/2})$ to $\hilb$ and $\theta_{i\ell}+\theta_{\ell i}$ is symmectric on this domain.
\end{lem}
\begin{proof}
We will prove continuity for $\theta:=\theta_{1 2}$. We multiply~\eqref{eq:T1def} by $\abs{p_2-k_{n+1}}^{D+\eps}$ and its inverse for any $\eps>0$, and use the Cauchy-Schwarz inequality to obtain
\begin{align*}
&\abs {\theta \hat \psi \uppar n}^2 
 \leq 
  \int_{\R^d} \frac{\abs{\hat v(q)}^2 \, \ud q}{L(P-{e}_2 q,\hat{K}_{n+1},q)\abs{p_2-q}^{2(D+\eps)}} 
\\
& 
\times \int_{\R^d}  \frac{\abs{\hat v(k_{n+1})}^2 \abs{\hat \psi\uppar n(P +(e_1-e_2) k_{n+1},\hat{K}_{n+1})}^2 \abs{p_2-k_{n+1}}^{2(D+\eps)}}{L(P-{e}_2 k_{n+1},K)} \, \ud k_{n+1}\, .
\end{align*}
Using the Hardy-Littlewood inequality and scaling, the integral in $q$ can be bounded by 
\begin{align*}
&\int_{\R^d}  \frac{\abs{ q}^{- 2\alpha}}{(p_1^2+q^2)\abs{q}^{2(D+ \eps)} }  \, \ud q
\leq C \abs{p_1}^{-(D+2\eps)} \, ,
\end{align*}
for $0<\eps<1/2$. Integrating in the remaining variables $(P, \hat K_{n+1})$ and performing a change of variables $P \rightarrow P+(e_1-e_2) k_{n+1}$ then gives
\begin{align*}
&\int \abs {\theta \hat \psi \uppar n (P,\hat{K}_{n+1})}^2 \, \ud \hat{K}_{n+1} \ud P
  \\
 &
\leq
C \int  \frac{\abs{ \hat v(k_{n+1})}^2 \abs{\hat \psi\uppar n(P,\hat{K}_{n+1})}^2 \abs{p_2}^{2(D+\eps)}}{L(P-{e}_1 k_{n+1},K) \abs{p_1-k_{n+1}}^{{D+2 \eps}} } \, \ud k_{n+1} \ud \hat{K}_{n+1} \ud P \, .
\end{align*}
Because $2+2\alpha+D+2 \eps = d+2\eps $ 
 the $k_{n+1}$-integral can, for $0<\eps<1$, be bounded as above by some constant times  $\abs{p_1}^{- 2\eps}$. We thus obtain
\begin{align*}
&\int \abs {\theta \hat \psi \uppar n (P,\hat{K}_{n+1}) }^2\, \ud \hat{K}_{n+1} \ud P
\leq
C   \int \abs{\hat{\psi}\uppar n(P,\hat{K}_{n+1})}^2 \abs{p_2}^{2D} \ud \hat{K}_{n+1} \ud P \, ,
\end{align*}
and this proves continuity.

To prove symmetry, we use the change of variables $Q=P+(e_i-e_\ell)k_{n+1}$ in
 \begin{align*}
  \langle &\hat \phi\uppar{n}, \theta_{i\ell}  \hat{\psi} \uppar n \rangle_{\hilb\uppar n}\\
  &= \int   \overline{\hat{\phi}\uppar{n}}(P,\hat K_{n+1}) \frac{\abs{\hat v(k_{n+1})}^2 \hat \psi\uppar n(P+({e}_i-{e}_\ell) k_{n+1},\hat{K}_{n+1})}{L(P-{e}_\ell k_{n+1},K) } \,  \ud P\ud K \\
  &=  \int   \overline{\hat{\phi}\uppar{n}}(Q+(e_\ell -e_i)k_{n+1}, \hat K_{n+1}),\hat K_{n+1}) \frac{\abs{\hat v(k_{n+1})}^2 \hat \psi\uppar n(Q,\hat{K}_{n+1})}{L(Q-{e}_i k_{n+1},K) } \,  \ud P\ud K.
   \end{align*}
%
%
Together with the bounds we have just proved, this implies that $\theta_{i \ell}^*$ extends $\theta_{\ell i}$ (defined on $D(L^{D/2})$), so the sum of the two is symmetric on this domain. 
\end{proof}
\begin{lem}
\label{lem:T2mainlem}
Assume Condition~\ref{cond:alphabeta} holds with $D \geq 0$ and let $u(s) = \frac{\beta }{2} s - \frac{D}{2}$. Then, for all $s > 0$ such that the following two conditions are satisfied
\begin{align}
\label{eq:cond1s}
u(s) &< 1
\\
\label{eq:cond2s}
0 &< u(u(s)) \, ,
\end{align}
and for all $i,\ell \in\{1,\dots, M\}$, the operator $\tau_{i \ell}$, defined in \eqref{eq:T2def}, is bounded from $ D(N^{\max(0,1-s)} L^{s-u(s)})$ to $\hilb$ and $\tau_{i\ell} + \tau_{\ell i}$ is symmetric on this domain.
\end{lem}
\begin{proof}
We start by proving the bound
\begin{equation*}
 \norm{\tau_{i \ell} \hat\psi\uppar{n}}_{\hilb \uppar n} \leq C n^{\max(0,1-s)} \norm{L^{s-u(s)} \psi\uppar{n}}_{\hilb\uppar{n}}
\end{equation*}
 for any fixed $i,\ell$ and $n\geq 1$ (note that $\tau_{i\ell}=0$ for $n=0$). 
 Note that, because $D\geq 0$ and $\beta \leq 2$, it holds that $u(s) \leq s$ and therefore the conditions~\eqref{eq:cond1s} and~\eqref{eq:cond2s} already imply that
\begin{align}
\label{eq:cond3s}
u(s),u(u(s)) \in (0,1) \, .
\end{align}
Now we denote $\tau=\tau_{i \ell}$ and write
\begin{align*}
&\tau \hat \psi \uppar n 
= 
\begin{aligned}[t]
\sum_{j=1}^{n} \int_{\R^d}& \omega(k_{n+1})^{\frac{s}{2}} \frac{ \hat v(k_{j}) \hat \psi\uppar n(P-{e}_\ell k_{n+1}+{e}_i k_j,\hat{K}_j)}{L(P-{e}_\ell k_{n+1},K)^{\frac{1}{2}} \omega(k_j)^{\frac{s}{2}}} 
\\
& \times \omega(k_j)^{\frac{s}{2}}  \frac{\overline{ \hat v(k_{n+1})}}{L(P-{e}_\ell k_{n+1},K)^{\frac{1}{2}} \omega(k_{n+1})^{\frac{s}{2}}}\, \ud k_{n+1} \, .
\end{aligned} 
\end{align*}
 Applying the Cauchy-Schwarz inequality on $\Lz(\R^d \times \lbrace 1, \dots, n \rbrace)$ and using the assumptions on $\hat v$ and $\omega$, we obtain
\begin{align*}
\abs{\tau \hat \psi \uppar n }^2 
&
\leq
\begin{aligned}[t]
&\sum_{j=1}^{n} \int_{\R^d} \omega(k_{n+1})^s \frac{ \abs{ \hat v(k_{j})}^2  \abs{\hat \psi\uppar n(P-{e}_\ell k_{n+1}+{e}_i k_j,\hat{K}_j)}^2}{L(P-{e}_\ell k_{n+1},K) \omega(k_j)^s} \, \ud k_{n+1} 
\\
&
 \times \sum_{\mu=1}^{n} \omega(k_\mu)^s  \int_{\R^d}  \frac{1}{((p_\ell-q)^2 + \Omega(\hat K_{n+1}))\abs{q}^{\beta s+2\alpha}} \, \ud q \,.
\end{aligned} 
\end{align*}
Since $u(s) \in (0,1)$, the integral in the second line is bounded by a constant times $\Omega(\hat{K}_{n+1})^{-u(s)}$.
In order to deal with the sum over $\mu=1,\dots, n$, we split the term $\mu=j$ from the rest and use~\eqref{eq:sum}. This gives 
\begin{align*}
\sum_{\mu=1}^{n} \omega(k_\mu)^s  \Omega(\hat{K}_{n+1})^{-u(s)}
&\leq
\omega(k_{j})^{s-u(s)}
+ n^{\max(0,1-s)} \Omega(\hat{K}_{n+1,j})^s \Omega(\hat{K}_{n+1})^{-u(s)}\\
&\leq \omega(k_{j})^{s-u(s)}
+ n^{\max(0,1-s)} \Omega(\hat{K}_{n+1,j})^{s-u(s)}\\
&\leq \omega(k_{j})^{s-u(s)}
+ n^{\max(0,1-s)} \Omega(\hat{K}_{j})^{s-u(s)}
\,,
 \end{align*}
 where we have  also used that $s\geq u(s)> 0$.
Consequently, we have a bound of the form
\begin{align*}
\abs{\tau \hat \psi \uppar n }^2 \leq  C \abs{\tau \uppar d \hat \psi \uppar n }^2 +C \abs{\tau \uppar {od} \hat \psi \uppar n }^2\,,  
\end{align*}
 with
\begin{align}\label{eq:tau_d}
&\abs{\tau \uppar d \hat \psi \uppar n }^2 
:= \sum_{j=1}^{n} \int_{\R^d}  \frac{\omega(k_{n+1})^s \abs{ \hat v(k_{j})}^2  \abs{\hat \psi\uppar n(P-{e}_\ell k_{n+1}+{e}_i k_j,\hat{K}_j)}^2}{\omega(k_j)^{u(s)}L(P-{e}_\ell k_{n+1},K) } \, \ud k_{n+1} 
\end{align}
and
\begin{align}\label{eq:tau_od}
&\abs{\tau \uppar {od} \hat\psi \uppar n }^2
:=  n^{\max(0,1-s)} 
\sum_{j=1}^{n} \int_{\R^d}
\begin{aligned}[t]
 &\frac{ \omega(k_{n+1})^s   \abs{\hat \psi\uppar n(P-{e}_\ell k_{n+1}+{e}_i k_j,\hat{K}_j)}^2}{L(P-{e}_\ell k_{n+1},K) }\\
 &\times \frac{\abs{ \hat v(k_{j})}^2\Omega(\hat{K}_{j})^{s-u(s)}}{\omega(k_j)^s}  \, \ud k_{n+1}\,.
\end{aligned}
\end{align}
To treat the term~\eqref{eq:tau_d}, we integrate in $(P,\hat K_{n+1})$, perform a change of variables $P \rightarrow P-{e}_\ell k_{n+1}+{e}_i k_j$, and then rename the variables $k_j \leftrightarrow k_{n+1}$. This yields
\begin{align*}
\int &\abs{\tau \uppar d \hat\psi \uppar n (P,\hat{K}_{n+1})}^2 \,  \ud P \ud \hat {K}_{n+1} \\
&
= \sum_{j=1}^{n} \int   \frac{ \omega(k_{n+1})^s\abs{ \hat v(k_{j})}^2  \abs{\hat  \psi\uppar n(P,\hat{K}_j)}^2}{\omega(k_j)^{u(s)}L(P-{e}_i k_{j},K) } \,  \ud P \ud {K} 
\\
&
= \sum_{j=1}^{n}  \int  \frac{ \omega(k_{j})^s\abs{ \hat v(k_{n+1})}^2  \abs{\hat  \psi\uppar n(P,\hat{K}_{n+1})}^2}{\omega(k_{n+1})^{u(s)}L(P-{e}_i k_{n+1},K) } \,  \ud P \ud K \, ,
\end{align*}
where, in the last step, we have used the permutation symmetry.
The $k_{n+1}$-integral can be estimated, using the assumptions on $\hat{v}$ and $\omega$ and the fact that $u(u(s)) \in (0,1)$, by
\begin{align*}
\int_{\R^d} \frac{\abs{ \hat v(k_{n+1})}^2  }{L(P-{e}_i k_{n+1},K) \omega(k_{n+1})^{u(s)}} \,  \ud k_{n+1}  \leq C \Omega(\hat{K}_{n+1})^{-u(u(s))} \, .
\end{align*}
Therefore, using again the bound~\eqref{eq:sum}, we conclude
\begin{align*}
\int &\abs{\tau \uppar d \hat\psi \uppar n (P,\hat{K}_{n+1})}^2 \,  \ud P \ud \hat {K}_{n+1} 
\\
&
\leq C \sum_{j=1}^{n}  \int  \omega(k_{j})^s \abs{\hat  \psi\uppar n(P,\hat{K}_{n+1})}^2 \Omega(\hat{K}_{n+1})^{-u(u(s))}  \,  \ud P \ud \hat{K}_{n+1} 
\\
&
\leq C n^{\max(0,1-s)} \int  \abs{\hat  \psi\uppar n(P,\hat{K}_{n+1})}^2 \Omega(\hat{K}_{n+1})^{s-u(u(s))}  \,  \ud P \ud \hat{K}_{n+1} \, .
\end{align*}
We proceed similarly with the second term~\eqref{eq:tau_od} and obtain
\begin{align*}
 &
\abs{\tau \uppar {od} \hat\psi \uppar n }^2
\leq
C n^{2 \max(0,1-s)}
\int
\abs{\hat  \psi\uppar n(P,\hat{K}_{n+1})}^2 \Omega(\hat{K}_{n+1})^{2 (s-u(s))}  \,  \ud P \ud \hat{K}_{n+1}  \, .
\end{align*}
This proves the desired bound, because $u(s)\leq s$ (as $D\geq 0$ and $\beta \leq 2$) and thus
\begin{align*}
s-u(u(s)) 
\leq s-u(s)+ u(s-u(s)) \leq  2(s-u(s)) \,.
\end{align*}

Symmetry follows from this as in Lemma~\ref{lem:DposT1}. In this case, the change of variables one makes is $P\mapsto P-e_\ell k_{n+1} + e_i k_j$. Additionally, one also uses the symmetry of functions in $\hilb\uppar{n}$, while renaming $k_j \leftrightarrow k_{n+1}$.
\end{proof}
%
%
 
\begin{rem}\label{rem:T}
An operator very similar to the operator $T$ plays an important role in the context of point interactions  of nonrelativistic particles, where $v=\delta$ and $\omega(k)=1+k^2$. This operator is known as the Ter-Martyrosyan--Skornyakov operator.

In two dimensions, this was studied in~\cite[Lem.3.1]{DeFiTe1994}, where estimates similar to ours (but with a linear growth in $n$) were proved. These bounds were refined by Griesemer and Linden~\cite{GrLi17}. 

The three-dimensional case has received more attention, see e.g.~\cite{DeFiTe1994, Co_etal15, MoSe17, MoSe18}. Recently, Moser and Seiringer~\cite{MoSe17} proved, in particular, an $n$-independent bound on $T_\mathrm{od}$ for this model, as an operator from $H^{1/2}(\R^{3+3n})$ to $H^{-1/2}(\R^{3+3n})$ (with $M=1$). Our proof of Lemma~\ref{lem:T2mainlem} is inspired by their technique.
\end{rem}
The lemmas above do provide bounds on $T$ for the case $d=3$, $v=\delta$, $\omega=k^2+1$ (for which $D=1$), as an operator on $D(L^{1/2})$. In particular, an $n$-independent bound on $T_\mathrm{od}$ on $H^{1}(\R^{3(M+n)})$ is obtained from Lemma~\ref{lem:T2mainlem} by choosing $s=1+\eps$. However, this model is not known to be renormalisable by Nelson's method and it does not satisfy the assumptions of Theorem~\ref{thm:main}. The reason is that, since $G$ does not map into $D(L^{1/2})$, we do not have $D(T)\subset G D(L)$ and  $D(H)\subset D(A)$. See~\cite{La18} for a modification of our method that works for this model.

\subsection{Proof of Theorem~\ref{thm:main} for $D \geq 0$}\label{sect:proof renorm}

We are now ready to prove Theorem~\ref{thm:main} under the assumptions of this section (Condition~\ref{cond:alphabeta},(2)).
As in the case of form perturbations treated in Section~\ref{sect:form}, we rewrite the Hamiltonian $H=L_0^*+ g A$ as (cf. Equation~\eqref{eq:H rewrite v2})
\begin{align*}
H 
&
= (1-G)^* L (1-G) + T \,. 
\end{align*}
From Lemma~\ref{lem:dressed} we already know that $H_0 := (1-G)^* L (1-G)$ is self-adjoint on $D(H_0)=D(H)=(1-G)^{-1}D(L)$. 
It is thus sufficient to prove that $T$ is symmetric and infinitesimally $H_0$-bounded on this domain. We will do this, distinguishing two cases.

\paragraph{The case $D=0$ and $\beta=2$.} 

In this case, Proposition~\ref{prop:Tmainprop} states that $T$ is symmetric on the domain $D(T)=D(L^\eps)$, for any $\eps>0$. 
Writing any $\psi\in D(H)$ as $(1-G)\psi + G\psi$, the first summand is an element of $D(L)$, and the second is in $D(L^\eps)$ by Corollary~\ref{lem:gonfockspace}.
We thus have $D(H)\subset D(L^\eps)=D(T)$ and $T$ is symmetric on $D(H)$.

To prove the relative bound on $T$, we decompose its action on $D(H)$ as $T=T(1-G)+TG$.
Because $G$ maps $\hilb$ to the domain of $T$, the operator $TG$ is bounded  on $\hilb$. To prove that $T(1-G)$ is relatively bounded by $H_0$ we simply use Young's inequality as in Equation~\eqref{eq:relboundifpowerbounded}.

\paragraph{The general case.} 
We will now cover the remaining cases, including the Nelson model. 
Given that $D$ and $\beta$ are within the bounds defined by Condition~\ref{cond:alphabeta},(2) the condition that either $\beta<2$ or $D>0$ is equivalent to $\beta-2<D$. We also recall from Equation~\eqref{eq:Dbeta} that Condition~\ref{cond:alphabeta},(2) implies
\begin{equation*}
 0\leq D< \frac{2\beta^2}{\beta^2+8}\leq \frac \beta 2
\end{equation*}
for the case at hand.

We will now use the flexibility of Proposition~\ref{prop:Tmainprop} that gives a family of domains on which $T$ is symmetric, by choosing a parameter $s(\beta, D)$ such that this domain is contained in $D(H)$.

\begin{lem}\label{lem:D_s}
 For any $s>0$ let $D_s(T)=D((N+1)^{\max(0,1-s)}L^{s-u(s)})$, with $u(s)=\tfrac\beta2 s - \tfrac D2$.
 If Condition~\ref{cond:alphabeta} is satisfied with $D\geq 0$, there exists $s=s(\beta,D)$, satisfying the conditions of Proposition~\ref{prop:Tmainprop}, and numbers $\delta_{1}(\beta,D),\delta_{2}(\beta,D)\in [0,1)$ such that
 \begin{itemize}
  \item $D(L^{\delta_1})\subset D_s(T)$, and
  \item $G$ is a continuous operator from $D(N^{\delta_2})$ to $D_s(T)$.
 \end{itemize}
\end{lem}
\begin{proof}
 For $\beta=2$, $D=0$ this was already proved above, so we may restrict to $\beta-2<D$. We will find $s$, depending on $\beta$ and $D$, such that the second statement holds. The first claim is then immediate, because
 \begin{equation*}
  (N+1)^{\max(0,1-s)}L^{s-u(s)} \leq \left\lbrace
  \begin{aligned}
  &(L+1)^{1-u(s)} \qquad &s\leq 1\\
  &L^{s-u(s)} \qquad &s> 1,
 \end{aligned}\right.
 \end{equation*}
and $u(s)>0$ (by the hypothesis $u(u(s))>0$ of Proposition~\ref{prop:Tmainprop}), as well as $s-u(s)<1/2$ (this follows from Equation~\eqref{eq:condlistnoreg} below since $\sigma-u(\sigma)>0$).

To prove the second claim, recall that, by Proposition~\ref{prop:generalboundong}, $G$ maps $\hilb\uppar{n}$ to $D(L^\eta)\cap \hilb\uppar{n+1}$, for an appropriate $\eta>0$ and any $n\in \N$. For $G$ to map into $D_s(T)$, we need to apply this with $\eta=s-u(s)$. If the hypothesis of Proposition~\ref{prop:generalboundong} are satisfied for some $\sigma\geq 0$, we then obtain the bound
\begin{equation*}
 \norm{G\psi}_{D_s(T)} \leq C \norm{(N+1)^{\frac{\max(0,1-\sigma)}2+\max(0,1-s)} \psi}_\hilb.
\end{equation*}


We will now prove the claim by showing that there is a possible choice of $(s,\sigma) \in (0, \infty) \times [0, \infty)$, satisfying the conditions of Proposition~\ref{prop:Tmainprop}, respectively Proposition~\ref{prop:generalboundong}, such that $\delta_2= \frac12{\max(0,1-\sigma)}+\max(0,1-s)$ is less than one. 

The parameter $\sigma$ needs to satisfy 
the hypothesis of Proposition~\ref{prop:generalboundong} with  $\eta=s-u(s)$:
 \begin{align}
  \label{eq:condlistsigma}
&u(\sigma)<1\,, \\
 \label{eq:condlistnoreg}
&s-u(s) + \frac{\sigma-u(\sigma)-1}{2}  < 0 \,.
 \end{align}
For $s$, the hypothesis of Proposition~\ref{prop:Tmainprop} have to hold:
\begin{align}
\label{eq:condlistus1}
&u(s)<1 \,,\\
\label{eq:condlistus2}
&u(u(s))>0\,.
 \end{align}
 
Set for $\beta<2$
\begin{align*}
S_1:= \frac{2 + D}{\beta} \,,  \qquad S_2 := \frac{1-\frac{3}{2}D}{2-\beta} \,,
\end{align*}
and $S_1=1+D/2$, $S_2=\infty$ for $\beta=2$. Note that $u(S_1)=1$ and $S_1 > 1$, because $\beta < D+2$. Furthermore, using that $D<\frac{2\beta^2}{\beta^2+8}$ and $0<\beta\leq 2$, we also have that
\begin{align}
\label{eq:boundonS2}
S_2=\frac{1-\frac{3}{2}D}{2-\beta} > \frac{1-\frac{3 \beta^2}{\beta^2+ 8}}{2-\beta} = \frac{\frac{\beta^2+ 8-3 \beta^2}{\beta^2+ 8}}{2-\beta} = 2 \frac{\frac{ 4- \beta^2}{\beta^2+ 8}}{2-\beta} = 2 \frac{ 2+ \beta}{\beta^2+ 8} > \frac{1}{2} \, .
\end{align}

We now define a family of pairs $(s_\eps ,  \sigma_\eps  )$ such that they fulfil the conditions \eqref{eq:condlistnoreg} -- \eqref{eq:condlistus2} as long as $\eps$ is small enough. For any $\eps>0$, let
\begin{equation*}
 ( s_\eps ,  \sigma_\eps  ):=
 \bigg(\min\lbrace S_1,S_2\rbrace-{\eps} \, , \ \max\left[0,\min\lbrace2(S_2- S_1),S_1\rbrace-2 \eps\right] \bigg).
\end{equation*}
For $\eps$ small enough, we can determine $( s_\eps ,  \sigma_\eps  )$ in all possible cases
\begin{align}
( s_\eps ,  \sigma_\eps  )
&
 = \begin{cases} \left(S_1-\eps,\min\lbrace2(S_2- S_1),S_1\rbrace-2 \eps\right) & S_1 < S_2 \\
 \left(S_2-\eps, 0\right) &  S_1 \geq S_2 \end{cases}
  \notag \\
 &
 = \begin{cases} \left(S_1-\eps,S_1-2 \eps\right) & \tfrac32 S_1 \leq S_2 \leq \infty
 \\ 
 \left(S_1-\eps,2(S_2- S_1)-2 \eps\right) & S_1 < S_2< \tfrac32 S_1 \\
 \left(S_2-\eps, 0\right) &  \tfrac12 < S_2\leq S_1\,.
 \end{cases} \label{eq:s cases}
\end{align}
In the last step we have used \eqref{eq:boundonS2}. Observe that $s_\eps$ and $\sigma_\eps$ are always finite, and ${s}_\eps > 0$, ${\sigma}_\eps\geq0$ if $\eps$ is small enough. 

It is clear from the definition that we have ${\sigma}_\eps < S_1$. Since $u$ is increasing for $\beta >0$ and $u(S_1)=1$, we conclude that $u({\sigma}_\eps) < u(S_1) = 1$, and~\eqref{eq:condlistsigma} holds. As also $s_\eps< S_1$, this equally shows that~\eqref{eq:condlistus1} is fulfilled.

To check~\eqref{eq:condlistus2}, observe that, because $\beta > D$,
\begin{align*}
u(u(S_1)) = u(1) = \frac{\beta}{2}- \frac{D}{2} > 0 \, .
\end{align*}
This shows~\eqref{eq:condlistus2} if $S_1\leq S_2$ and $\eps$ is small enough. If $S_2< S_1$, then necessarily $\beta<2$ and, using the hypothesis $D<\frac{2\beta^2}{\beta^2+8}$, 
we see that
\begin{align*}
u(u(S_2)) 
&
= \frac{\beta^2}{4} \frac{1-\tfrac32 D}{2-\beta} -   \frac{(2+{\beta})D}{4}
> \frac{\beta^2}{4(2-\beta)} \bigg(1- \frac{3\beta^2}{\beta^2+8} - \frac{2(4-\beta^2)}{\beta^2+8}\bigg)
=0\,  .
\end{align*}
This proves that~\eqref{eq:condlistus2} holds for any sufficiently small $\eps$.

The last condition to prove is~\eqref{eq:condlistnoreg}. By computing $2s_\eps + \sigma_\eps$ in the different cases of Equation~\eqref{eq:s cases}, we find
\begin{align*}
2  s_\eps +  \sigma_\eps =  
\begin{cases}  3 S_1 - 4 \eps   &\tfrac32 S_1 \leq S_2 \leq \infty 
\\
2 S_2 - 4 \eps  & S_1 < S_2< \tfrac32 S_1
\\
 2 S_2-2 \eps  &  \tfrac12 < S_2\leq S_1
 \, . 
 \end{cases}
\end{align*}
From this we see that $2  s_\eps +  \sigma_\eps<2S_2$, and thus
\begin{align*}
s_\eps-u(s_\eps)-\frac{1}{2}\left(1-\sigma_\eps+u(\sigma_\eps) \right)
&
= 
\frac{1}{2} \left(1-\frac{\beta}{2}\right) (2 s_\eps +\sigma_\eps) +\frac{3 D}{4} - \frac{1}{2} 
\nonumber \\
&
<
\frac{1}{2} \left(1-\frac{3}{2} D\right) + \frac{3}{4} D - \frac{1}{2} = 0 \, ,
\end{align*} 
which proves~\eqref{eq:condlistnoreg}.

It remains to compute $\delta_2=\max(0,1-s_\eps) + \frac{1}{2} \max(0,1- \sigma_\eps) $ 
and see that $\delta_2<1$.
Since, for $\eps$ small enough, $S_1-\eps>1$, we find for the different cases of Equation~\eqref{eq:s cases}
\begin{align*}
\delta_2
&
= 
\begin{cases}
0 & \tfrac32 S_1 \leq S_2 \leq \infty
\\ 
\tfrac12 \max(0,1-2(S_2- S_1)+2\eps) & S_1<S_2<\tfrac32 S_1
\\
\max(0,1-S_2+\eps)+ \tfrac12 &  \tfrac12<S_2\leq S_1\,.
\end{cases}
\end{align*}
In the first case, we are finished. In the second case, $S_2-S_1>0$ and choosing $\eps$ smaller than this quantity proves the claim. For the last one, it is sufficient to choose $\eps<S_2-\tfrac12$, which is positive by~\eqref{eq:boundonS2}. This completes the proof.
\end{proof}

This lemma proves that $D(H)$ in $D_s(T)$, because $\psi=(1-G)\psi +G\psi$, with both of these terms in $D_s(T)$ since $D(H)\subset D(N)$ by Lemma~\ref{lem:Dposnisleftinv}. Since $\delta_1, \delta_2<1$, the lemma also implies that $(T, D_s(T))$ is infinitesimally $H_0$-bounded, because $N$ is $H_0$-bounded by Equation~\eqref{eq:NisrelboundedbyH0}.
We have thus proven that $H$ is self-adjoint and bounded from below under the assumptions of Theorem~\ref{thm:main}. 

The expression~\eqref{eq:H create}, involving the creation operators, for $H$ as an operator from $D(H)$ to the dual of $D(L)$ was already derived in Equation~\eqref{eq:L+a^*}. Note that $A\psi \in \mathscr{H}$ for $\psi\in D(H)$, since $T$ maps $(1-G)^{-1}D(L)$ to $\hilb$, as we have just shown.

\subsection{Proof of Theorem~\ref{thm:renorm}}\label{sect:proof limit}

We will now prove that the operator $H$, whose self-adjointness was proved in the previous section, is equal to an operator $H_\infty$ constructed by renormalisation.

Let us recall the definition of $H_\infty$. Let,  for $\Lambda>0$, $v_\Lambda$ be the interaction defined by $\hat{v}_\Lambda(k)=\chi_\Lambda(k) \hat v(k)$, where $\chi_\Lambda$ is the characteristic function of a ball with radius $\Lambda$. Then let
\begin{align*}
H_\Lambda 
= L + g \sum_{i=1}^M a(v_\Lambda(x_i-y)) + a^*(v_\Lambda(x_i-y))  \, .
\end{align*}
Since $v_\Lambda \in L^2(\R^d)$, this operator is self-adjoint on the domain $D(H_\Lambda)=D(L)$. 
In order to consider the limit of $H_\Lambda$ as $\Lambda\to \infty$ it is necessary to modify it by adding  
\begin{align*}
E_\Lambda:=  {g^2 M} \int_{\R^d} \frac{ \abs{\hat v_\Lambda(k)}^2 }{k^2+\omega(k) }  \, \ud k \,.
\end{align*}  
 Note that, since we are assuming that the second case of Condition~\ref{cond:alphabeta} holds, the numbers $E_\Lambda$ diverge as $\Lambda \to \infty$.
 
 It is known that, under appropriate assumptions on $\hat v$ and $\omega$, the limit as $\Lambda \to \infty$ of $H_\Lambda + E_\Lambda$ exists (see~\cite[Thm 3.3]{GrWu17}).
 
 \begin{thm*}[\cite{nelson1964, GrWu17}]
  Let Condition~\ref{cond:alphabeta} be satisfied with $D\geq 0$. Then $H_\Lambda +E_\Lambda$ converges in norm resolvent sense as $\Lambda\to \infty$ to an operator $(H_\infty,D(H_\infty))$ that is self-adjoint.
 \end{thm*}

We will now prove Theorem~\ref{thm:renorm}, which states that under the same hypothesis $H_\Lambda +E_\Lambda$ converges to $H$ in the strong resolvent sense. This obviously implies $H=H_\infty$. With a more involved analysis one could certainly also prove convergence of the resolvents in norm. However, this seems unnecessary as the main point is to show that  $H=H_\infty$, and this already implies norm resolvent convergence by~\cite[Thm 3.3]{GrWu17}.

 In the following proof, an important role will be played by $G$ and its regularised variant $G_\Lambda := -g L^{-1} a^*(V_\Lambda)$. The operators $(1-G_\Lambda)$ are somewhat analogous to the Gross transformation $U_\Lambda$ that is used in the renormalisation procedure. This is a family of unitary operators on $\hilb$ with the property that $H_\Lambda+ E_\Lambda=U^*_\Lambda (L + R_\Lambda) U_\Lambda$, with operators $R_\Lambda$ that have a limit as $\Lambda\to \infty$, in the sense of quadratic forms on $D(L^{1/2})$. The limit $\lim_{\Lambda\to \infty} U_\Lambda=:U_\infty$ also exists and one has
 \begin{equation*}
H_\infty=  U^*_\infty (L \dotplus B_\infty) U_\infty \, ,
\end{equation*}
where $L \dotplus B_\infty$ denotes the self-adjoint operator defined by the sum of the quadratic forms.
This implies that $D(\abs{H_\infty}^{1/2})=U_\infty^* D(L^{1/2})$. However, for an explicit characterisation of $D(H_\infty)$ one would need to know the domain of $L \dotplus B_\infty$ and  an explicit description of the action of $U_\infty$ on this domain.
On the other hand, using the operators $G_\Lambda$ and $G_\infty=G$, we will find that
\begin{equation*}
 H_\Lambda + E_\Lambda = (1-G_\Lambda)^* L (1-G_\Lambda) + T_\Lambda + E_\Lambda\,.
\end{equation*}
The operators $T_\Lambda + E_\Lambda$ will converge as $\Lambda \to \infty$ to $T$ (strongly as operators $D(T)\to \hilb$). We have shown, in Section~\ref{sect:proof renorm}, that $T$ is a perturbation of $(1-G)^* L (1-G)$ in the sense of operators, and thus $D(H)=(1-G)^{-1}D(L)$.  

While these procedures look rather similar, there are some notable differences. The Gross transformation is constructed as a Weyl operator from the one-particle function $\hat{v}_\Lambda(k)/(k^2+\omega(k))$, it is unitary and maps the form domain of $H_\Lambda$, respectively $H_\infty$, to $D(L^{1/2})$. On the other hand, the operator $1-G$ uses the resolvent of the multi-particle operator $L$ and it is invertible, but not unitary. 
Like $U_\infty$, this operator maps $D(\abs{H}^{1/2})$ to $D(L^{1/2})$ (see also~\eqref{eq:H form dom}), but additionally also $D(H)$ to $D(L)$.
The action of $(1-G)$ on a generic element of $\hilb$ is also somewhat easier to analyse. This is because $\left((1-G)\psi\right)\uppar{n}$ depends only on $\psi\uppar{n}$ and $\psi\uppar{n-1}$, whereas $\left(U_\infty\psi\right)\uppar{n}$ will depend on all of the $\psi\uppar{j}$, $j\in \N$.

 \begin{proof}[Proof of Theorem~\ref{thm:renorm}]
%
Let $a(V_\Lambda)= \sum_{i=1}^M a(v_\Lambda(x_i-y))$, define $G_\Lambda = -g L^{-1} a^*(V_\Lambda)$, and 
\begin{equation*}
 T_\Lambda:=- G_\Lambda^* L G_\Lambda =- g^2 a(V_\Lambda)L^{-1} a(V_\Lambda)^*\,.
\end{equation*}
Since $v_\Lambda \in \Lz$ for $\Lambda < \infty$ and $L\geq N$, one easily sees that $G_\Lambda$ and $T_\Lambda$ are bounded operators on $\hilb$. We then have
\begin{align*}
(1-G_\Lambda)^* L (1-G_\Lambda) + T_\Lambda 
&
=
 L - G_\Lambda^* L - L G_\Lambda + G_\Lambda^* L G_\Lambda + T_\Lambda
 \\
 &
 = L +  g \left( a(V_\Lambda) +a^*(V_\Lambda)  \right)\\
 &=H_\Lambda\,.
 \end{align*}
 Using this representation, we calculate the difference of resolvents
 \begin{align}
(H+&\ui)^{-1} - (H_\Lambda + E_\Lambda +\ui)^{-1}
\nonumber 
\\
=&
 (H+\ui)^{-1}\Big(H_\Lambda+E_\Lambda- H \Big)(H_\Lambda + E_\Lambda +\ui)^{-1}
 \nonumber
 \\
 =&(H+\ui)^{-1}
 \Big((1-G)^*L(G-G_\Lambda) \Big) (H_\Lambda + E_\Lambda +\ui)^{-1} 
 \label{eq:H conv1}\\
 &+(H+\ui)^{-1}
 \Big((G^*-G^*_\Lambda)L(1-G_\Lambda) \Big) (H_\Lambda + E_\Lambda +\ui)^{-1} 
 \label{eq:H conv2}\\
 &
 +
 (H+\ui)^{-1}\Big(T_\Lambda +E_\Lambda - T\Big)(H_\Lambda + E_\Lambda +\ui)^{-1}\,.
 \label{eq:H conv3}
\end{align}
We need to prove that this converges to zero, strongly on $\hilb$. 

Consider first 
\begin{equation*}
 G-G_\Lambda=g L^{-1}\left( a^*(V_\Lambda)-a^*(V) \right)=g L^{-1}\left( \sum_{i=1}^M a^*\Big((v_\Lambda-v)(x_i-y)\Big)\right).
\end{equation*}
Following the proof of Proposition~\ref{prop:generalboundong}, with $\hat v$ replaced by $\hat{v} (\chi_\Lambda-1)$,  one easily sees that this converges to zero, since integrals such as~\eqref{eq:G int bound} tend to zero with the modified interaction. This proves the convergence of the term~\eqref{eq:H conv1}, because $T$ is $H$-bounded, as shown in Section~\ref{sect:proof renorm}, and thus $(H+\ui)^{-1}(1-G)^*L$ is bounded. 
The proof of this statement, with $v$ replaced by $v_\Lambda$ can also be used to show that $T_\Lambda +E_\Lambda$ is bounded relative to $(1-G_\Lambda)^*L(1-G_\Lambda)$ with constants independent of $\Lambda$, because all of the estimates are given by certain integrals of $\hat v_\Lambda$ that are bounded by the integral with $\hat v$ (see also the discussion of $T_\Lambda$ below). This implies that $L(1-G_\Lambda)(H_\Lambda + E_\Lambda +\ui)^{-1} $ is bounded uniformly in $\Lambda$ and gives the desired result for~\eqref{eq:H conv2}.

We now turn to $T_\Lambda + E_\Lambda = T_{\ud, \Lambda} + E_\Lambda + T_{\uod, \Lambda} $, with $T_{\ud, \Lambda}$, $T_{\uod, \Lambda}$ defined in analogy with $T_\ud$, $T_\uod$ (see Equations~\eqref{eq:DposDefofTd},~\eqref{eq:Toffdiagdef v1}). In Fourier representation the action of $T_{\ud, \Lambda} + E_\Lambda $ is just multiplication by the function
\begin{align*}
- {g^2} \sum_{\ell=1}^M & \int_{\abs{k_{n+1}}< \Lambda} \abs{\hat v(k_{n+1})}^2   \left(\frac{ 1 }{L(P-{e}_\ell k_{n+1},K) } - \frac{ 1 }{k_{n+1}^2+\omega(k_{n+1}) } \right) \, \ud k_{n+1}\,.
\end{align*}
As $\Lambda \to \infty$ this converges to the function defining $T_\ud$, given in~\eqref{eq:Dposregintergal}, pointwise. Using the bound of Lemma~\ref{lem:DposTdiag} one then sees that $T_{\ud, \Lambda} + E_\Lambda \to T_\ud$ in the strong topology of operators from $D(L^{\max(\eps,D/2)})$ to $\hilb$.

Concerning $T_{\uod, \Lambda}$, we spell out the action of $g a(V_\Lambda)G_\Lambda$ in the same way as in \eqref{eq:Tformally v1} and decompose as in \eqref{eq:Toffdiagdef v1} to arrive at
\begin{align*}
T_{\uod, \Lambda}- T_{\uod} := -g^2  \sum_{\ell=1}^M \sum_{i=1,i\neq \ell}^M (\theta_{i \ell, \Lambda} - \theta_{i \ell}) -g^2 \sum_{\ell=1}^M \sum_{i=1}^M (\tau_{i \ell, \Lambda} - \tau_{i \ell}) \, .
\end{align*}
Explicitly, we have
\begin{align}
\label{eq:T1lambda}
(\theta_{i \ell, \Lambda}& - \theta_{i \ell}) \hat{\phi} \uppar n (P,\hat{K}_{n+1}) 
\nonumber \\
&
 = \int_{\R^d}  \frac{(\chi_\Lambda(k_{n+1})-1)\abs{\hat v(k_{n+1})}^2 \hat \psi\uppar n(P+({e}_i-{e}_\ell) k_{n+1},\hat{K}_{n+1})}{L(P-{e}_\ell k_{n+1},K) } \, \ud k_{n+1}\,,  
\end{align}
and 
\begin{align}
\label{eq:T2lambda}
 &(\tau_{i \ell, \Lambda} - \tau_{i \ell})\hat{\phi} \uppar n (P,\hat{K}_{n+1})  
 %
 = \sum_{j=1}^{n} \int_{\R^d}
 \begin{aligned}[t]
  &\frac{\overline{\hat  v(k_j)} \hat  v(k_{n+1})  \hat \psi\uppar n(P-{e}_\ell k_{n+1}+{e}_i k_j,\hat{K}_j)}{L(P-{e}_\ell k_{n+1},K) } \\
  &\times\left(\chi_\Lambda(k_j) \chi_\Lambda(k_{n+1})-1\right)\, \ud k_{n+1}\,.
 \end{aligned}
\end{align}
With the expression \eqref{eq:T1lambda} at hand, going through the proof of Lemma~\ref{lem:DposT1} shows that $\sum_{\ell=1}^M \sum_{i=1,i\neq \ell}^M (\theta_{i \ell, \Lambda} - \theta_{i \ell})$ converges to zero strongly as an operator from $D(L^{D/2})$ to $\hilb$.
To show the analogue for the $\tau$-terms, one first inserts the equality
\begin{equation*}
 \chi_\Lambda(k_j) \chi_\Lambda(k_{n+1})-1= \chi_\Lambda(k_j)(\chi_\Lambda(k_{n+1})-1) +(\chi_\Lambda(k_j)-1)  
\end{equation*}
into~\eqref{eq:T2lambda}. Then, one observes that at least one of the the integrals in $k_j$ or $k_{n+1}$ performed in the proof of Lemma~\ref{lem:T2mainlem} converges to zero. This implies that~\eqref{eq:T2lambda}
converges to zero strongly as an operator from $D(N^{\max(0,1-s)} L^{s-u(s)})$ to $\hilb$.

To summarise, we have found that $T_\Lambda+ E_\Lambda - T$ tends to zero  strongly as an operator from $D(T)$ to $\hilb$, for any domain $D(T)$ that can be chosen in Proposition~\ref{prop:Tmainprop}. Combining this with the fact that $T$ is bounded relative to $H$ implies that for any $\psi\in \hilb$
\begin{equation*}
 \lim_{\Lambda\to \infty} \norm{\left(T_\Lambda+ E_\Lambda - T\right) (H-\ui)^{-1}\psi}_{\hilb}=0\,.
\end{equation*}
This shows convergence of~\eqref{eq:H conv3} and completes the proof. 
\end{proof}

\section{Regularity of domain vectors}\label{sect:regularity}
In this section we will discuss the regularity of vectors in $D(H)$. 
These results apply both to the case of form perturbations of Section~\ref{sect:form} and the renormalisable models treated in Section~\ref{sect:renorm}. 
Due to the boundary condition $(1-G)\psi\in D(L)$, a vector $\psi\in D(H)$  is exactly as regular as $G\psi = \psi - (1-G)\psi$ is. The same reasoning also applies on the form domain of $H$. Since we proved in Sections~\ref{sect:proof form} and~\ref{sect:proof renorm} that $H$ is a perturbation of $H_0=(1-G)^*L(1-G)$, the quadratic form of $H$ is a perturbation of that of $H_0$ and its domain is
\begin{equation}\label{eq:H form dom}
D(\abs {H}^{1/2})=(1-G)^{-1}D(L^{1/2}) \subset D(N^{1/2})\,.
\end{equation} 
This domain is characterised by the abstract boundary condition $\psi-G\psi \in D(L^{1/2})$, which is non-trivial if $G\psi \notin D(L^{1/2})$, i.e. for the models treated in Section~\ref{sect:renorm}. In this case, $\psi\in D(\abs {H}^{1/2})$ has the same regularity (with respect to $L$) as $G\psi$. 

We will prove sharp results on the regularity of $G\psi$ below. Together, these will imply the Corollary~\ref{cor:Nelson} stated in the introduction.


Proposition~\ref{prop:generalboundong} establishes that if $\abs{\hat v(k)}\leq \abs{k}^{-\alpha}$, then the vectors in the domain of the operator $H$ with interaction $v$ have the regularity of those in $D(L^\eta)$ for all $\eta<\frac{2-D}4=1-\frac{d-2\alpha}4$. Note that if $\int \frac{ |\hat v(k)|^2}{k^2+\omega(k)} \ud k<\infty$ Condition~\ref{cond:alphabeta} implies that we are in the case of form perturbations with $D<0$ treated in Section~\ref{sect:form} and the following corollary holds for some $\eta>1/2$.

\begin{cor}\label{cor:reg}
Let the conditions of Theorem~\ref{thm:main} be satisfied. Then for every $0\leq \eta < \frac{2-D}4$ we have
\begin{equation*}
D(H)\subset D(L^\eta) \quad \text{and} \quad  D(\abs{H}^{1/2})\subset D(L^{\min(\eta,1/2)}).
\end{equation*}

\end{cor}
\begin{proof} 
Let $\psi \in D(H)$, respectively $\psi \in  D(\abs{H}^{1/2})$. 
 To show that $G\psi \in D(L^\eta)$ we can apply Proposition~\ref{prop:generalboundong} with $s=0$, since $\eta - \frac{2-D}4 <0$. This yields
 \begin{equation*}
  \norm{L^\eta G\psi}_\hilb \leq C \|\sqrt{N+1} \psi\|_\hilb\,.
 \end{equation*}
Together with the fact that $D(\abs{H}^{1/2})\subset D(N^{1/2})$ this implies that $G\psi \in D(L^\eta)$.
\end{proof}

For the Fröhlich model this means that $D(H)\subset D(L^\eta)$ for $\eta<3/4$. For the Nelson model as well as our model for point-particles in two dimensions with $v=\delta$ we have $D(\abs{H}^{1/2})\subset D(L^\eta)$ for $\eta<1/2$.

We will now show that these results are sharp, in the sense that $D(H)\cap D(L^\eta)=\{0\}$ for all larger $\eta$.
The intuition behind this is that the (worst) singularities of $G\psi$ behave exactly like those of $(-\Delta + \omega(-\ui \nabla))^{-1}v(x_\ell-y_{j})$. 
Similar results for $M=1$ were also proved in~\cite{GrWu16} and~\cite{GrWu17} using the Gross transform.

\begin{prop}\label{prop:non-reg}
Assume the hypothesis of Theorem~\ref{thm:main} hold and additionally that $\omega \in L^\infty_\mathrm{loc}(\R^d)$. Let $0<\eta<1$ be such that $\int \frac{ |\hat v(k)|^2}{(k^2+\omega(k))^{2(1-\eta)}} \ud k=\infty$, then
 \begin{equation*}
D(H)\cap D(L^\eta)=\{0\},
 \end{equation*}
 and if $\eta\leq 1/2$ we  also have
 \begin{equation*}
  D(\abs{H}^{1/2})\cap D(L^\eta)=\{0\}\,.
 \end{equation*}
\end{prop}

\begin{proof}
We will show that $G$ maps no $0 \neq \psi \in  \hilb$ into $D(L^\eta)$, which implies our claim as discussed above. 

Let $n\in \N$ be such that $\psi\uppar{n}\neq 0$ and recall that
 \begin{align*}
  \widehat{G\psi\uppar{n}}(P,K)=\frac{-g}{\sqrt{n+1}} \sum_{i=1}^M\sum_{j=1}^{n+1} L^{-1}(P,K) \hat v(k_j) \hat\psi\uppar{n}(P+e_i k_j, \hat K_j).
 \end{align*}
 Let $U\subset \R^{Md}\times \R^{(n+1)d}$ be the set 
 \begin{equation*}
  U=\{(P,K): |p_j|<R \text{ and } \abs{k_j}<R \text{ for all } j>1 \}= \R^d\times B_R(0)^{M-1} \times \R^d \times B_R(0)^{n}\,,
 \end{equation*}
 where $R>0$ is a parameter, to be chosen later. We will prove that
 \begin{equation*}
  \int_U \abs{\widehat{L^\eta G\psi\uppar{n}}(P,K)}^2 \, \ud P \ud K = \infty\,,
 \end{equation*}
which implies that $G \psi\uppar{n}\notin D(L^\eta)$. We first use that $(a+b)^2\geq \tfrac12 a^2- b^2$ and the Cauchy-Schwarz inequality to obtain the lower bound
\begin{align}
 \label{eq:L^sG decomp 1}
 \abs{\widehat{L^\eta G\psi\uppar{n}}(P,K)}^2
 \geq& \frac{g^2}{2(n+1)} \frac{\abs{\hat v(k_1)}^2 \abs{\hat\psi\uppar{n}(P+e_1 k_1, \hat K_1)}^2}{L(P,K)^{2-2\eta}}\\
 &- g^2 M \sum_{(i,j)\neq (1,1)} \frac{\abs{\hat v(k_j)}^2\abs{\hat \psi\uppar{n}(P+e_i k_j, \hat K_j)}^2}{L(P,K)^{2-2\eta}}.
 \label{eq:L^sG decomp 2}
\end{align}
We will see that the terms of the second line have a finite integral over $U$, while the integral of the first is infinite.
In the sum over $(i,j)\neq (1,1)$ in Equation~\eqref{eq:L^sG decomp 2} consider a term with $i=1, j> 1$.  By the change of variables $p_1\mapsto q=p_1-k_j$ (note that the domain of integration for $p_1$ is $\R^d$) we obtain the bound
\begin{equation*}
 \int\limits_U \frac{\abs{\hat v(k_j)}^2\abs{\hat \psi\uppar{n}(P+e_1 k_j, \hat K_j)}^2}{L(P,K)^{2-2\eta}} \, \ud P \ud K 
 \leq \norm{\psi\uppar{n}}^2 \sup_{q\in \R^d} \int\limits_{\abs{k}<R}\hspace{-6pt} \frac{\abs{\hat v(k)}^2}{((q-k)^2+1)^{2-2\eta}} \, \ud k \, .
\end{equation*}
This is finite since $\hat v \in L^2_\mathrm{loc}$. The terms with $i,j>1$ can be bounded by enlarging the domain of integration in the variable $p_i$ to $\R^d$ and proceeding as for $i=1$.
The terms with $j=1, i>1$ are estimated in the same way, where the change of variables is performed in $k_1$ and the remaining integral is then over $p_i$.

To show that the integral over the term~\eqref{eq:L^sG decomp 1} is infinite, we perform a change of variables $p_1 \rightarrow p_1-k_1$. Then we restrict the domain of integration to $\{\abs{p_1}<R\} \cap U$ to bound it from below by
\begin{align}
\label{eq:stillwithL}
  \int\limits_{B_R(0)^{M+n}} \abs{\hat \psi\uppar{n}(P, \hat K_1)}^2 \int\limits_{\R^d} \frac{\abs{\hat v(k_1)}^2 }{ L(p_1-k_1, P,K) ^{2\eta-2}} \, \ud k_1 \ud P \ud \hat K_1\, .
\end{align}
Since we have restricted to $(P,\hat{K}_1) \in B_R(0)^{M+n}$
 and assumed that $\omega \in L^\infty_\mathrm{loc}$, it holds that $P^2+\Omega(\hat{K}_1) \leq C$ for some $C > 0$ that depends on $R$. Because, in particular $\abs{p_1} < R$ and $1\leq \omega(k_1)$, we can then estimate
\begin{align*}
 L(p_1-k_1, \hat{P}_1,K) 
 &
 \leq (k_1-p_1)^2 + \omega(k_1)  + C
  \leq C' ( k_1^2 + \omega(k_1)) \, ,
\end{align*}
for some $C' > 0$. Hence the integral~\eqref{eq:stillwithL} is bounded from below by some constant times
\begin{align*}
\int\limits_{B_R(0)^{M+n}} \abs{\hat \psi\uppar{n}(P, \hat K_1)}^2 \, \ud P \ud \hat K_1 \int_{\R^d} \frac{\abs{\hat v(k_1)}^2 }{(k_1^2 + \omega(k_1))^{2\eta-2}} \, \ud k_1  \, .
\end{align*}
Because $\psi \uppar n \neq 0$, we can choose an $R>0$ such that 
\begin{equation*}
 \int\limits_{B_R(0)^{M+n}} \abs{\hat \psi\uppar{n}(P, \hat K_1)}^2 \, \ud P \ud K > 0\,.
\end{equation*}
But since the integral in $k_1$ is infinite by hypothesis we have proved the claim. 
\end{proof}

\paragraph{Acknowledgments.} 
We thank Stefan Keppeler, Stefan Teufel and Roderich Tumulka for helpful discussions. J.S. 
was supported by the German Research Foundation (DFG) within the Research Training Group 1838 \textit{Spectral Theory and Dynamics of Quantum Systems} and thanks the Laboratoire Interdisciplinaire Carnot de Bourgogne for its hospitality during his stay in Dijon.

\begin{appendix}
\section{Appendix}

\begin{lem}\label{lem:V}
 Let $v\in \mathscr{S}'(\R^d)$ with $\hat v \in L^\infty(\R^d)+L^2(\R^d)$. If $v\notin L^2(\R^d)$, then for any $M\geq 2$ the set of $X=(x_1,\dots, x_M)$ such that 
 \begin{equation*}
V(X,y)=\sum_{j=1}^M v(x_j -y)    
 \end{equation*}
is an element of $L^2(\R^d, \ud y)$ has Lebesgue measure zero in $\R^{M d}$.
\end{lem}
\begin{proof}
 Assume to the contrary that the set where $V(X,y)\in L^2(\R^d, \ud y)$ has positive measure. 
 This set is the union over all $R>0$ of the sets
  \begin{equation*}
  U_R:=\{ X \in \R^{Md}: \abs{X}<R \text{ and } \norm{V(X,y)}_{L^2(\R^d)} < R \}\,,
 \end{equation*}
 and thus $U_R$ has positive (and finite) measure $\abs{U_R}>0$ for some $R>0$.  Integrating over this set, we see that $\int_{U_R} V(X,x_1-y) \ud X \in L^2(\R^d)$, since
\begin{equation*}
 \norm{\int_{U_R} V(X,x_1-y) \ud X}_{L^2(\R^d)} \leq \int_{U_R} \norm{V(X,y)}_{L^2(\R^d)} \ud X < R \abs{U_R}\,.
\end{equation*}
On the other hand, denoting by $\chi_{U_R}$ the characteristic function of $U_R$, we have
\begin{align*}
 \int_{U_R} &V(X,x_1-y) \ud X\\
 &= \abs{U_R} v(y) + \sum_{j=2}^M  \int_{U_R} v(x_j-x_1+y) \ud X\\
 &= \abs{U_R} v(y) + \sum_{j=2}^M \int_{\R^d}  v(y-x_1) \int_{\R^{(M-1)d}} \chi_{U_R} (X+e_1 x_j) \ud X\,.
\end{align*}
Since $U_R$ has finite measure, the functions
\begin{equation*}
 f_j (x) = \int_{\R^{(M-1)d}} \chi_{U_R} (X+e_1 x_j) \ud x_2 \cdots \ud x_M
\end{equation*}
are in $L^1\cap L^2(\R^d)$. Since $\hat v\in L^2+L^\infty$, the convolution $v\ast f_j$ is then in $L^2(\R^d)$.
But this implies that 
\begin{equation*}
 \abs{U_R} v(y) = \int_{U_R} V(X,x_1-y) \ud X - \sum_{j=2}^M (v\ast f_j) (y)\in L^2(\R^d)\,, 
\end{equation*}
 a contradiction.
\end{proof}

\begin{lem}
\label{lem:parameterint}
Let $p \in \R^3$ and $\theta \in (1,3)$. Then there exists a constant $C > 0$ such that
\begin{align}
\label{eq:parameterintdrei}
\int_{\R^3} \frac{ \ud q}{((p-q)^2+1) \abs{q}^\theta} \leq \frac{C}{\abs{p}^{\theta-1}} \, .
\end{align}
Let $p \in \R^2$ and $\theta \in \lbrace 1, 2 \rbrace$. Then there exists a constant $C>0$ such that 
\begin{align}
\label{eq:parameterintzwei}
\int_{\R^2} \frac{ \ud q}{((p-q)^2+1) (q^2+1)^{\frac{\theta}{2}}} \leq \frac{C}{\abs{p}^{\theta}} (\log\left(1+{\abs p}\right) +1) \, .
\end{align}
\begin{proof}
For $d=3$, we will use spherical coordinates and write $p$ instead of $\abs p$ when it is clear what is intended. Let $R>0$ be any positive number. For $p\geq R$ we have:
\begin{align*}
\int_{\R^3} \frac{ \ud q}{((p-q)^2+1) \abs{q}^\theta} 
&
= 
2 \pi \int_0^\infty \frac{1}{r^{\theta-2}} \int_{-1}^1 \frac{1}{r^2 - 2 r  p s + p^2 + 1} \, \ud s \ud r
\\
&
=
  \frac{\pi}{p } \int_0^\infty \frac{1}{r^{\theta-1} } \log\left(\frac{r^2 + 2 r p  + p^2 + 1}{r^2 - 2 r  p  + p^2 + 1} \right) \, \ud r
\end{align*}
We perform a change of variables $r \rightarrow \frac{r}{p}=:t$ which yields
\begin{align*}
&  \frac{\pi}{p^{\theta-1} } \int_0^\infty \frac{1}{t^{\theta-1} } \log\left(\frac{(t  + 1)^2 + p^{-2} }{(t-1)^2 + p^{-2}} \right) \, \ud t \, .
\end{align*}
Now we split the domain of integration into three parts. For $t \in [0,\frac{1}{2}]$ we will use that
\begin{align}
&  \frac{(t  + 1)^2 + p^{-2} }{(t-1)^2 + p^{-2}} = 1 + \frac{4 t}{(t-1)^2 + p^{-2}} \, ,
\nonumber
\end{align}
which implies
\begin{align}
\label{eq:logest}
&\log\left(\frac{(t  + 1)^2 + p^{-2} }{(t-1)^2 + p^{-2}} \right) \leq \frac{4 t}{(t-1)^2 + p^{-2}} \, .
\end{align}
So we estimate the integral there using $\theta<3$ by
\begin{align*}
& \frac{4 \pi}{p^{\theta-1} } \int_0^{\frac{1}{2}} \frac{1}{t^{\theta-2} }\frac{1 }{(t-1)^2 + p^{-2}} \,   \ud t 
\leq
 \frac{4 \pi}{p^{\theta-1} } \int_0^{\frac{1}{2}} \frac{1}{t^{\theta-2} }\frac{1 }{(t-1)^2 } \,   \ud t \leq \frac{C}{p^{\theta-1} } \, .
 \end{align*}
 We can do the exact same thing for $t \in [\frac{3}{2},\infty)$ and obtain the same bound. It remains to deal with the integrable singularity at $t=1$:
 \begin{align*}
\frac{\pi}{p^{\theta-1} } &\int_{\frac{1}{2}}^{\frac{3}{2}} \log\left(\frac{(t+1)^2 + p^{-2} }{(t-1)^2+ p^{-2}} \right)  \ud t  
\\
&
\leq
\frac{\pi}{p^{\theta-1} } \int_{\frac{1}{2}}^{\frac{3}{2}} \log\left({(t+1)^2 + R^{-2} }\right)  \ud t  
- \frac{\pi}{p^{\theta-1} } \int_{\frac{1}{2}}^{\frac{3}{2}} \log\left({(t-1)^2} \right)  \ud t \leq \frac{C}{p^{\theta-1} } \, .
\end{align*}
If however $p<R$, then we simply estimate the integral by a constant. Since $R$ was arbitrary, this yields the claim in the case $d=3$. 

If $d=2$, we first observe that for $a \in (0,1)$ it holds that
\begin{align}
\label{eq:antidervcos}
\frac{\partial}{\partial x} \arctan\left(\sqrt{\frac{a+1}{a-1}} \tan(x) \right) = \frac{\sqrt{a^2 -1}}{1-a \cos(2 x)} \, .
\end{align}
We will use this to integrate in the angular variable:
\begin{align*}
\int_{\R^2}& \frac{ \ud q}{((p-q)^2+1) (q^2+1)^{\frac{\theta}{2}}}
= 
\int_0^\infty \frac{r }{ (r^2+1)^{\frac{\theta}{2}}} \int_{-\pi}^\pi \frac{1}{r^2 - 2 r  p \cos(s) + p^2 + 1} \, \ud r \ud s
\\
&
\leq
\int_0^\infty \frac{1}{ (r^2+1)^{\frac{\theta-1}{2}}(r^2+p^2+1)} \int_{-\frac{\pi}{2}}^{\frac{\pi}{2}} \frac{1}{1- \frac{2 r  p}{r^2+p^2+1} \cos(2 \xi)} \,  \ud r \ud \xi
\end{align*}
Now set $a=\frac{2 r  p}{r^2+p^2+1}$ and use \eqref{eq:antidervcos} to obtain
\begin{align*}
&
\int_0^\infty  \frac{a }{ 2 r p (r^2+1)^{\frac{\theta-1}{2}}}  \left[ \frac{1}{ \sqrt{a^2 -1} } \arctan\left(\sqrt{\frac{a+1}{a-1}} \tan(\xi) \right) \right]_{-\frac{\pi}{2}}^{\frac{\pi}{2}} \, \ud r
\\
&
= 
\int_0^\infty \frac{\pi }{(r^2+1)^{\frac{\theta-1}{2}}  ((r-p)^2+1)^{\frac{1}{2}}((r+p)^2+1)^{\frac{1}{2}}} \, \ud r \, .
\end{align*}
We then perform a change of variables  $r \rightarrow \frac{r}{p}=:x$. This yields 
\begin{align*}
&
\frac{\pi}{p^{\theta}}\int_0^\infty \frac{1}{ (x^2+p^{-2})^{\frac{\theta-1}{2}}  ((x-1)^2+p^{-2})^{\frac{1}{2}}((x+1)^2+p^{-2})^{\frac{1}{2}}} \, \ud x \, .
\end{align*}
For $\theta=1$ the integral has one singularity at $x=1$. For $\theta=2$ there are two singularities remaining, one at zero and another one at $x=1$. For that reason we split the integral at $r=\frac{1}{2}$ and $r=\frac{3}{2}$ (as in the case $d=3$ above). The integral from $r=\frac{3}{2}$ to infinity is finite and bounded independent of $p$. For the other two terms we use the fact that $  \mathrm{arsinh}(x)' = (x^2+1)^{-1/2}$ and conclude that
\begin{align*}
\int\limits_0^\infty \frac{ \ud x }{ (x^2+p^{-2})^{\frac{\theta-1}{2}}  ((x-1)^2+p^{-2})^{\frac{1}{2}}((x+1)^2+p^{-2})^{\frac{1}{2}}} 
\leq
C_1 \mathrm{arsinh}\left( \frac{\abs p}{2}\right) + C_2 \, ,
\end{align*}
for some constants $C_1,C_2>0$. Choose $C=\pi \max(C_1,C_2)$ and note that $\mathrm{arsinh}(x) = \log(x+\sqrt{x^2+1}) \leq \log(2 x + 1)$ for non-negative $x$. This yields the claim.
\end{proof}
\end{lem}
\end{appendix}


\newcommand{\etalchar}[1]{$^{#1}$}
\providecommand{\bysame}{\leavevmode\hbox to3em{\hrulefill}\thinspace}
\providecommand{\MR}{\relax\ifhmode\unskip\space\fi MR }
\providecommand{\MRhref}[2]{%
  \href{http://www.ams.org/mathscinet-getitem?mr=#1}{#2}
}
\providecommand{\href}[2]{#2}

\end{document}